\documentclass[a4paper,USenglish,cleveref,thm-restate]{lipics-v2021}

\usepackage{amsmath,amssymb}
\usepackage{stmaryrd}
\usepackage{etoolbox}
\usepackage{mathtools}
\usepackage{mathpartir}
\usepackage{upgreek}
\usepackage{prooftree}
\usepackage{xspace}

\newif\ifcomments
\commentstrue
\commentsfalse

\newcommand{\ie}{\textit{i.e.}\xspace}

\newcommand{\mktag}[1]{\mathsf{\color{teal}#1}}
\newcommand{\mkkeyword}[1]{\mathsf{\color{blue}#1}}
\newcommand{\mkfunction}[1]{\mathsf{#1}}

\newcommand{\parens}[1]{(#1)}
\newcommand{\braces}[1]{\{#1\}}
\newcommand{\bracks}[1]{[#1]}
\newcommand{\angles}[1]{\langle#1\rangle}

\newcommand{\set}[1]{\braces{#1}}

\newcommand{\rulename}[1]{\textnormal{\textsc{\small[#1]}}}

\newcommand{\defrule}[2][]{\hypertarget{rule:\ifblank{#1}{#2}{#1}}{\rulename{#2}}}
\newcommand{\refrule}[2][]{\hyperlink{rule:\ifblank{#1}{#2}{#1}}{\rulename{#2}}}

\newcommand{\Nat}{\mathbb{N}}
\newcommand{\muMALL}{$\mu\textsf{\upshape MALL}^\infty$\xspace}
\newcommand{\piLIN}{\texorpdfstring{$\pi\textsf{LIN}$}{piLIN}\xspace}
\newcommand{\CP}{$\textsf{CP}$\xspace}
\newcommand{\muCP}{$\mu\textsf{CP}$\xspace}
\newcommand{\piDILL}{$\pi\textsf{DILL}$\xspace}
\newcommand{\picalculus}{\texorpdfstring{$\pi$-calculus}{pi-calculus}\xspace}

\newcommand\eoe{\hfill\ensuremath{\blacksquare}}

\newcommand{\marginnote}[2]{%
  \ifcomments%
    $^{\color{magenta}\mathclap\star}$%
    \marginpar[
        \flushright\tiny\sf\textbf{#1}: #2
    ]{
        \flushleft\tiny\sf\textbf{#1}: #2
    }%
  \fi%
}

\newcommand{\LP}[1]{\marginnote{LP}{\color{blue}#1}}

\newcommand{\REVISION}[1]{#1}

\newcommand\Run{\rho}

\newcommand{\InTag}{\mktag{in}}
\newcommand{\LeftTag}{\InTag_1}
\newcommand{\RightTag}{\InTag_2}
\newcommand{\AddTag}{\mktag{add}}
\newcommand{\PayTag}{\mktag{pay}}
\newcommand{\PlayTag}{\mktag{play}}
\newcommand{\QuitTag}{\mktag{quit}}
\newcommand{\WinTag}{\mktag{win}}
\newcommand{\LoseTag}{\mktag{lose}}

\newcommand{\x}{x}
\newcommand{\y}{y}
\newcommand{\z}{z}

\newcommand{\out}[1]{\overline{#1}}
\newcommand{\inp}[1]{#1}
\newcommand{\parop}{\mathbin\|}

\newcommand{\Call}[2]{#1\ifblank{#2}{}{\angles{#2}}}
\newcommand{\Link}[2]{#1\leftrightarrow#2}
\newcommand{\Close}[1]{\out{#1}\parens{}}
\newcommand{\Wait}[1]{\inp{#1}\parens{}}
\newcommand{\Fail}[1]{\mkkeyword{case}\,\inp{#1}\braces{}}
\newcommand{\Select}[3][]{#2\,\out{#3}\ifblank{#1}{}{\parens{#1}}}
\newcommand{\Left}[2][]{\Select[#1]{\InTag_1}{#2}}
\newcommand{\Right}[2][]{\Select[#1]{\InTag_2}{#2}}
\newcommand{\CaseX}[4][]{\mkkeyword{case}\,\inp{#2}\ifblank{#1}{}{\parens{#1}}\braces{#3,#4}}
\newcommand{\Case}[4][]{\CaseX[#1]{#2}{#3}{#4}}
\newcommand{\Fork}[5][]{\out{#2}\parens{#3\ifblank{#1}{}{,#1}}\parens{#4\parop#5}}
\newcommand{\FreeFork}[3]{\out{#1}\angles{#2,#3}}
\renewcommand{\Join}[3][]{\inp{#2}\parens{#3\ifblank{#1}{}{,#1}}}
\newcommand{\Cut}[3]{\parens{#1}\parens{#2\parop#3}}
\newcommand{\Choice}[3][]{#2 \oplus\ifblank{#1}{}{_{#1}} #3}
\newcommand{\Rec}[2][]{\mkkeyword{rec}\,\out{#2}\ifblank{#1}{}{\parens{#1}}}
\newcommand{\Corec}[2][]{\mkkeyword{corec}\,\inp{#2}\ifblank{#1}{}{\parens{#1}}}

\newcommand{\AddressSet}{\mathcal{A}}
\newcommand{\address}{\addressA}
\newcommand{\addressA}{\alpha}
\newcommand{\addressB}{\beta}

\newcommand{\Formulas}{\Phi}
\newcommand{\Formula}{\FormulaF}
\newcommand{\FormulaF}{\varphi}
\newcommand{\FormulaG}{\psi}

\newcommand{\X}{X}
\newcommand{\Y}{Y}

\newcommand{\mkformula}[1]{#1}

\newcommand{\Bot}{\mkformula\bot}
\newcommand{\Top}{\mkformula\top}
\newcommand{\One}{\mkformula{\mathbf{1}}}
\newcommand{\Zero}{\mkformula{\mathbf{0}}}
\newcommand{\choice}{\mathbin{\mkformula\oplus}}
\newcommand{\branch}{\mathbin{\mkformula\binampersand}}
\newcommand{\tfork}{\mathbin{\mkformula\otimes}}
\newcommand{\tjoin}{\mathbin{\mkformula\bindnasrepma}}
\newcommand{\tmu}{\mkformula\mu}
\newcommand{\tnu}{\mkformula\nu}

\newcommand{\RankSet}{\Nat^{\infty}}
\newcommand{\rank}{\rankR}
\newcommand{\rankR}{r}
\newcommand{\rankS}{s}

\newcommand{\Context}{\ContextC}
\newcommand{\ContextC}{\Upgamma}
\newcommand{\ContextD}{\Updelta}

\newcommand{\qtp}[3][]{#3 \vdash\ifblank{#1}{}{_{\color{red}XXX#1}} #2}
\newcommand{\wtp}[3][]{#3 \Vdash\ifblank{#1}{}{_{\color{red}XXX#1}} #2}

\newcommand{\LinkRule}{ax}
\newcommand{\CutRule}{cut}
\newcommand{\FailRule}{$\Top$}
\newcommand{\CloseRule}{$\One$}
\newcommand{\WaitRule}{$\Bot$}
\newcommand{\ForkRule}{$\tfork$}
\newcommand{\JoinRule}{$\tjoin$}
\newcommand{\SelectRule}{$\choice$}
\newcommand{\CaseRule}{$\branch$}
\newcommand{\RecRule}{$\tmu$}
\newcommand{\CorecRule}{$\tnu$}
\newcommand{\ChoiceRule}{choice}

\newcommand{\pcong}{\preccurlyeq}
\newcommand{\red}{\rightarrow}

\newcommand{\wred}{\Rightarrow}

\newcommand{\eqdef}{\stackrel{\smash{\textsf{\upshape\tiny def}}}=}

\newcommand{\prefix}{\sqsubseteq}
\newcommand{\subf}{\preceq}
\newcommand{\tred}{\leadsto}

\newcommand{\fn}[1]{\mkfunction{fn}\parens{#1}}
\newcommand{\bn}[1]{\mkfunction{bn}\parens{#1}}
\newcommand{\subst}[2]{\braces{#1/#2}}
\newcommand{\dual}[1]{#1^{\bot}}
\newcommand{\dom}[1]{\mkfunction{dom}\parens{#1}}
\newcommand{\InfOften}[1]{\mkfunction{inf}\parens{#1}}
\newcommand{\minf}{\mkfunction{min}\,}
\newcommand{\strip}[1]{\overline{#1}}
\newcommand{\depth}[1]{\mathsf{depth}\parens{#1}}
\newcommand{\rankof}[1]{|#1|}
\newcommand{\fc}[2]{\mathsf{fc}(#1,#2)}

\theoremstyle{remark}
\newtheorem{notation}[theorem]{Notation}

\title{An Infinitary Proof Theory of Linear Logic Ensuring Fair Termination in the Linear \texorpdfstring{$\pi$}{pi}-Calculus}
\titlerunning{Fair Termination in the Linear \texorpdfstring{$\pi$}{pi}-Calculus}

\author{Luca Ciccone}{Universit\`a di Torino,
  Italy}{luca.ciccone@unito.it}{https://orcid.org/0000-0001-9515-5280}{}
\author{Luca Padovani}{Universit\`a di Torino,
  Italy}{luca.padovani@unito.it}{https://orcid.org/0000-0001-9097-1297}{}

\authorrunning{L. Ciccone and L. Padovani}
\Copyright{Luca Ciccone and Luca Padovani}

\ccsdesc[500]{Theory of computation~Linear logic}
\ccsdesc[500]{Theory of computation~Process calculi}
\ccsdesc[300]{Theory of computation~Program analysis}

\keywords{Linear \picalculus, Linear Logic, Fixed Points, Fair Termination}

\acknowledgements{We are grateful to Simona Ronchi Della Rocca and Francesco
  Dagnino for their comments on an early draft of this paper, to Stephanie
  Balzer for having provided pointers to related work and to the anonymous
  CONCUR reviewers for their questions and detailed feedback.}

\EventEditors{Bartek Klin, S\l{}awomir Lasota, and Anca Muscholl}
\EventNoEds{3}
\EventLongTitle{33rd International Conference on Concurrency Theory (CONCUR 2022)}
\EventShortTitle{CONCUR 2022}
\EventAcronym{CONCUR}
\EventYear{2022}
\EventDate{September 12--16, 2022}
\EventLocation{Warsaw, Poland}
\EventLogo{}
\SeriesVolume{243}
\ArticleNo{16}

\nolinenumbers

\begin{document}
\maketitle

\begin{abstract}
    Fair termination is the property of programs that may diverge ``in
    principle'' but that terminate ``in practice'', \ie under suitable fairness
    assumptions concerning the resolution of non-deterministic choices. We study
    a conservative extension of \muMALL, the infinitary proof system of the
    multiplicative additive fragment of linear logic with least and greatest
    fixed points, such that cut elimination corresponds to fair termination.
    Proof terms are processes of \piLIN, a variant of the linear $\pi$-calculus
    with (co)recursive types into which binary and (some) multiparty sessions
    can be encoded. As a result we obtain a behavioral type system for \piLIN
    (and indirectly for session calculi through their encoding into \piLIN) that
    ensures fair termination: although well-typed processes may engage in
    arbitrarily long interactions, they are \emph{fairly} guaranteed to
    eventually perform all pending actions.
\end{abstract}

\section{Introduction}
\label{sec:introduction}

The \emph{linear $\pi$-calculus}~\cite{KobayashiPierceTurner99} is a typed
refinement of Milner's $\pi$-calculus in which linear channels can be used only
once, for a one-shot communication. As it turns out, the linear $\pi$-calculus
is the fundamental model underlying a broad family of communicating processes.
In particular, all \emph{binary
sessions}~\cite{Honda93,HondaVasconcelosKubo98,HuttelEtAl16} and some
\emph{multiparty sessions}~\cite{HondaYoshidaCarbone16} can be encoded into the
linear
$\pi$-calculus~\cite{Kobayashi02b,CairesPerez16,DardhaGiachinoSangiorgi17,CicconePadovani20}.
Sessions are private communication channels linking two or more processes and
whose usage is disciplined by a \emph{session type}, a type representing a
structured communication protocol. In all session type systems, session
endpoints are \emph{linearized channels} that can be used repeatedly but in a
sequential manner. The key insight for encoding sessions into the linear
$\pi$-calculus is to represent linearized channels as chains of linear channels:
when some payload is transmitted over a linear channel, it can be paired with
another linear channel on which the subsequent interaction takes place.

In this work we propose a type system for \piLIN, a linear $\pi$-calculus with
(co)recursive types, such that well-typed processes are \emph{fairly
terminating}. Fair termination~\cite{GrumbergFrancezKatz84,Francez86} is a
liveness property stronger than \emph{lock
freedom}~\cite{Kobayashi02,Padovani14} -- \ie the property that every pending
action can be eventually performed -- but weaker than \emph{termination}. In
particular, a fairly terminating program may diverge, but all of its infinite
executions are considered ``unfair'' -- read impossible or unrealistic -- and so
they can be ignored insofar termination is concerned. A simple example of fairly
terminating program is that modeling a repeated interaction between a buyer and
a seller in which the buyer may either pay the seller and terminate or may add
an item to the shopping cart and then repeat the same behavior. In principle,
there is an execution of the program in which the buyer keeps adding items to
the shopping cart and never pays. In practice, this behavior is considered
unfair and the program terminates under the fairness assumption that the buyer
eventually pays the seller.

Our type system is a conservative extension of
\muMALL~\cite{BaeldeDoumaneSaurin16,Doumane17,BaeldeEtAl22}, the infinitary
proof system for the multiplicative additive fragment of linear logic with least
and greatest fixed points. In fact, the modifications we make to \muMALL are
remarkably small: we add one (standard) rule to deal with
\emph{non-deterministic choices}, those performed autonomously by a process, and
we relax the validity condition on \muMALL proofs so that it only considers the
``fair behaviors'' of the program it represents.
The fact that there is such a close correspondence between the typing rules of
\piLIN and the inference rules of \muMALL is not entirely surprising. After all,
there have been plenty of works investigating the relationship between
$\pi$-calculus terms and linear logic proofs, from those of
Abramsky~\cite{Abramsky94}, Bellin and Scott~\cite{BellinScott94} to those on
the interpretation of linear logic formulas as session
types~\cite{DeYoungCairesPfenningToninho12,Wadler14,CairesPfenningToninho16,LindleyMorris16,RochaCaires21,QianKavvosBirkedal21}.
%
% These interpretations provide sound logical foundations for models of
% communicating processes by establishing a tight correspondence between cut
% elimination and communication. As a consequence, well-known results in the
% domain of linear logic -- most notably cut elimination -- translate for free
% into appealing properties of well-typed processes (deadlock and lock
% freedom~\cite{Wadler14,CairesPfenningToninho16} or
% termination~\cite{LindleyMorris16}).
%
Nonetheless, we think that the connection between \piLIN and \muMALL stands out
for two reasons.
First, \piLIN is conceptually simpler and more general than the session-based
calculi that can be encoded in it. In particular, all the session calculi based
on linear logic rely on \REVISION{an asymmetric interpretation of the
multiplicative connectives $\tfork$ and $\tjoin$ so that} $\FormulaF \tfork
\FormulaG$ (respectively, $\FormulaF \tjoin \FormulaG$) is the type of a session
endpoint used for sending (respectively, receiving) a message of type
$\FormulaF$ and then used according to $\FormulaG$. In our setting, the
connectives $\tfork$ and $\tjoin$ retain their symmetry since we interpret
$\FormulaF \tfork \FormulaG$ and $\FormulaF \tjoin \FormulaG$ formulas as the
output/input of pairs, in the same spirit of the original encoding of linear
logic proofs proposed by Bellin and Scott~\cite{BellinScott94}. This
interpretation gives \piLIN the ability of modeling \emph{bifurcating protocols}
of which binary sessions are just a special case.
The second reason why \piLIN and \muMALL get along has to do with the cut
elimination result for \muMALL. In finitary proof systems for linear logic, cut
elimination may proceed by removing \emph{topmost cuts}. In \muMALL there is no
such notion as a topmost cut since \muMALL proofs may be infinite. As a
consequence, the cut elimination result for \muMALL is proved by eliminating
\emph{bottom-most cuts}~\cite{BaeldeEtAl22}. This strategy fits perfectly with
the reduction semantics of {\piLIN} -- and that of any other conventional
process calculus, for that matter -- whereby reduction rules act only on the
exposed (\ie unguarded) part of processes but not behind prefixes. As a result,
the reduction semantics of \piLIN is completely ordinary, unlike other
logically-inspired process calculi that incorporate commuting conversions
\cite{Wadler14,LindleyMorris16} or other unnatural prefix
swapping~\cite{BellinScott94}.

In previous work~\cite{CicconePadovani22} we have proposed a type system
ensuring the fair termination of binary sessions. The present work achieves the
same objective using a more basic process calculus and exploiting its strong
logical connection with \muMALL. In fact, the soundness proof of our type system
piggybacks on the cut elimination property of \muMALL.
%
% Also, there exist simple processes that are well typed in this work but not in
% the previous one (\cref{ex:forwarder}). 
%
Other session typed calculi based on linear logic with fixed points have been
studied by Lindley and Morris~\cite{LindleyMorris16} and by Derakhshan and
Pfenning~\cite{DerakhshanPfenning19,Derakhshan21}.
The type systems described in these works respectively guarantee termination and
strong progress, whereas our type system guarantees fair termination which is
somewhat in between these properties. Overall, our type system seems to hit a
sweet spot: on the one hand, it is deeply rooted in linear logic and yet it can
deal with common communication patterns (like the buyer/seller interaction
described above) that admit potentially infinite executions and therefore are
out of scope of other logic-inspired type systems; on the other hand, it
guarantees lock freedom~\cite{Kobayashi02,Padovani14}, strong
progress~\cite[Theorem 12.3]{DerakhshanPfenning19} and also termination, under a
suitable fairness assumption.

The paper continues as follows.
\cref{sec:language} presents \piLIN and the fair termination property ensured by
our type system. \cref{sec:types} describes \piLIN types, which are suitably
embellished \muMALL formulas. \cref{sec:type-system} describes the inference
rules of \muMALL rephrased as typing rules for \piLIN. \cref{sec:soundness}
identifies the valid typing derivations and states the properties of well-typed
processes. \cref{sec:related} discusses related work in more detail and
\cref{sec:conclusion} presents ideas for future developments. Supplementary
material, additional examples and proofs can be found in the appendix.

%%% Local Variables:
%%% mode: latex
%%% TeX-master: "main"
%%% End:

\newcommand\Buyer{\textit{Buyer}}
\newcommand\Seller{\textit{Seller}}

\section{Syntax and Semantics of \piLIN}
\label{sec:language}

In this section we define syntax and reduction semantics of \piLIN, a variant of
the linear $\pi$-calculus~\cite{KobayashiPierceTurner99} in which all channels
are meant to be used for \emph{exactly} one communication. The calculus supports
(co)recursive data types built using units, pairs and disjoint sums. These data
types are known to be the essential ingredients for the encoding of sessions in
the linear
$\pi$-calculus~\cite{Kobayashi02b,DardhaGiachinoSangiorgi17,ScalasDardhaHuYoshida17}.

\begin{table}
    \caption{\label{tab:syntax} Syntax of \piLIN.}
    \begin{math}
        \displaystyle
        \begin{array}[t]{@{}r@{~}c@{~}ll@{}}
            P, Q
            & ::= & \Link\x\y & \text{link}
            \\
            & | & \Fail\x & \text{empty input}
            \\
            & | & \Wait\x.P & \text{unit input}
            \\
            & | & \Join[y]\x\z.P & \text{pair input}
            \\
            & | & \Case[y]\x{P}{Q} & \text{sum input}
            \\
            & | & \Corec[y]\x.P & \text{corecursion}
        \end{array}
        ~
        \begin{array}[t]{@{}r@{~}c@{~}lll@{}}
            & | & \Cut\x{P}{Q} & \text{composition}
            \\
            & | & \Choice{P}{Q} & \text{choice}
            \\
            & | & \Close\x & \text{unit output}
            \\
            & | & \Fork[y]\x\z{P}{Q} & \text{pair output}
            \\
            & | & \Select[y]{\InTag_i}\x.P & \text{sum output} & i\in\set{1,2}
            \\
            & | & \Rec[y]\x.P & \text{recursion}
        \end{array}
    \end{math}
\end{table}

We assume given an infinite set of \emph{channels} ranged over by $x$, $y$ and
$z$. \piLIN processes \REVISION{are coinductively generated by the productions
of the grammar} shown in \cref{tab:syntax} and their informal meaning is given
below.
A \emph{link} $\Link\x\y$ acts as a \emph{linear
forwarder}~\cite{GardnerLaneveWischik07} that forwards a single message either
from $x$ to $y$ or from $y$ to $x$. The uncertainty in the direction of the
message is resolved once the term is typed and the polarity of the types of $x$
and $y$ is fixed (\cref{sec:type-system}).
The term $\Fail\x$ represents a process that receives an empty message from $x$
and then fails. \REVISION{This form is only useful in the metatheory: the} type
system guarantees that well-typed processes never fail, since it is not possible
to send empty messages.
The term $\Close\x$ models a process that sends the unit on $x$, effectively
indicating that the interaction is terminated, whereas $\Wait\x.P$ models a
process that receives the unit from $x$ and then continues as $P$.
The term $\Fork[z]\x\y{P}{Q}$ models a process that creates two new channels $y$
and $z$, sends them in a pair on channel $x$ and then forks into two parallel
processes $P$ and $Q$. Dually, $\Join[z]\x\y.P$ models a process that receives a
pair containing two channels $y$ and $z$ from channel $x$ and then continues as
$P$.
The term $\Select[y]{\InTag_i}\x.P$ models a process that creates a new channel
$y$ and sends $\InTag_i\parens\y$ (that is, the $i$-th injection of $y$ in a
disjoint sum) on~$x$. Dually, $\Case[y]\x{P_1}{P_2}$ receives a disjoint sum
from channel $x$ and continues as either $P_1$ or $P_2$ depending on the tag
$\InTag_i$ it has been built with. For clarity, in some examples we will use
more descriptive labels such as $\AddTag$ and $\PayTag$ instead of $\LeftTag$
and $\RightTag$.
The terms $\Rec[y]\x.P$ and $\Corec[y]\x.P$ model processes that respectively
send and receive a new channel $y$ and then continue as $P$. They do not
contribute operationally to the interaction being modeled, but they indicate the
points in a program where (co)recursive types are unfolded.
A term $\Cut\x{P}{Q}$ denotes the parallel composition of two processes $P$ and
$Q$ that interact through the fresh channel $x$.
Finally, the term $\Choice{P}{Q}$ models a non-deterministic choice between two
behaviors $P$ and $Q$.

\piLIN binders are easily recognizable because they enclose channel names in
round parentheses. Note that all outputs are in fact \emph{bound outputs}. The
output of free channels can be modeled by combining bound outputs with
links~\cite{LindleyMorris16}. \REVISION{For example, the output
$\FreeFork\x\y\z$ of a pair of free channels $y$ and $z$ can be modeled as the
term $\Fork[z']\x{y'}{\Link\y{y'}}{\Link\z{z'}}$.}
We identify processes modulo renaming of bound names, we write $\fn{P}$ for the
set of channel names occurring free in $P$ and we write $\subst\y\x$ for the
capture-avoiding substitution of $y$ for the free occurrences of $x$.
\REVISION{We impose a well-formedness condition on processes so that, in every
sub-term of the form $\Fork[z]\x\y{P}{Q}$, we have $y \not\in\fn{Q}$ and
$z\not\in\fn{P}$.}

We omit any concrete syntax for representing infinite processes. Instead, we
work directly with infinite trees obtained by corecursively unfolding
contractive equations of the form $A(x_1,\dots,x_n) = P$. For each such
equation, we assume that $\fn{P} \subseteq \set{x_1,\dots,x_n}$ and we write
$\Call{A}{y_1,\dots,y_n}$ for its unfolding $P\subst{y_i}{x_i}_{1\leq i\leq n}$.
\cref{sec:infinite-processes} describes the changes for supporting a more
conventional (but slightly heavier) handling of infinite processes.

\begin{notation}
    \label{not:sessions}
    To reduce clutter due to the systematic use of bound outputs, by convention
    we omit the continuation called $y$ in \cref{tab:syntax} when its name is
    chosen to coincide with that of the channel $x$ on which $y$ is
    sent/received.
    For example, with this notation we have $\Join\x\z.P = \Join[x]\x\z.P$ and
    $\Select{\InTag_i}\x.P = \Select[x]{\InTag_i}\x.P$ and $\Case\x{P}{Q} =
    \Case[x]\x{P}{Q}$.
    \eoe
\end{notation}

A welcome side effect of adopting \cref{not:sessions} is that it gives the
illusion of working with a session calculus in which the same channel $x$ may be
used repeatedly for multiple input/output operations, while in fact $x$ is a
linear channel used for exchanging a single message along with a fresh
continuation that turns out to have the same name. In fact, if one takes this
notation as native syntax for a session calculus, its linear $\pi$-calculus
encoding~\cite{DardhaGiachinoSangiorgi17} turns out to be precisely the \piLIN
term it denotes.
Besides, the idea of rebinding the same name over and over is widespread in
session-based functional languages \cite{GayVasconcelos10,Padovani17} as it
provides a simple way of ``updating the type'' of a session endpoint after each
use.

\begin{example}
    \label{ex:buyer-seller}
    Below we model the interaction informally described in
    \cref{sec:introduction} between buyer and seller using  the syntactic sugar
    defined in \cref{not:sessions}:
    \[
        \Cut\x{\Call\Buyer\x}{\Call\Seller{x,y}}
        \text{\quad where\quad}
        \begin{array}{@{}r@{~}l@{}}
            \Buyer(x) & =
            \Rec\x.\parens{
                \Choice{
                    \Select\AddTag\x.\Call\Buyer\x
                }{
                    \Select\PayTag\x.\Close\x
                }
            }
            \\
            \Seller(x,y) & =
            \Corec\x.
            \Case\x{
                \Call\Seller{x,y}
            }{
                \Wait\x.\Close\y
            }
        \end{array}
    \]

    At each round of the interaction, the buyer decides whether to $\AddTag$ an
    item to the shopping cart and repeat the same behavior (left branch of the
    choice) or to $\PayTag$ the seller and terminate (right branch of the
    choice). The seller reacts dually and signals its termination by sending a
    unit on the channel $y$. \REVISION{As we will see in \cref{sec:type-system},
    $\Rec\x$ and $\Corec\x$ identify the points within processes where
    (co)recursive types are unfolded.}

    If we were to define $\Buyer$ using distinct bound names we would write an
    equation like
    \[
        \Buyer(x) =
        \Rec[y]\x.
        \parens{
            \Choice{
                \Select[z]\AddTag\y.\Call\Buyer\z
            }{
                \Select[z]\PayTag\y.\Close\z
            }
        }
    \]
    and similarly for $\Seller$. 
    \eoe
\end{example}

\begin{table}
    \caption{\label{tab:semantics}Structural pre-congruence and reduction semantics of \piLIN.}
    \begin{math}
        \displaystyle
        \begin{array}{@{}lr@{~}c@{~}ll@{}}
            \defrule{s-link} &
            \Link\x\y & \pcong & \Link\y\x
            \\
            \defrule{s-comm} &
            \Cut\x{P}{Q} & \pcong & \Cut\x{Q}{P}
            \\
            \defrule{s-assoc} &
            \Cut\x{P}{\Cut\y{Q}{R}} & \pcong & \Cut\y{\Cut\x{P}{Q}}{R} &
            \text{if $x\in\fn{Q}$}
            \\
            & & & & \text{and $y\not\in\fn{P}$}
            % \\
            % \defrule{s-call} &
            % \Call{A}{\seqof\x} & \pcong & P & \Let{A}{\seqof\x}{P}
            \\
            \\
            \defrule{r-link} &
            \Cut\x{\Link\x\y}{P} & \red & P\subst\y\x
            \\
            \defrule{r-unit} &
            \Cut\x{\Close\x}{\Wait\x.P} & \red & P
            \\
            \defrule{r-pair} &
            \Cut\x{\Fork[y]\x\z{P_1}{P_2}}{\Join[y]\x\z.Q} & \red & \Cut\z{P_1}{\Cut\y{P_2}{Q}}
            \\
            \defrule{r-sum} &
            \Cut\x{\Select[y]{\InTag_i}\x.P}{\Case[y]\x{P_1}{P_2}} & \red & \Cut\y{P}{P_i} & i\in\set{1,2}
            \\
            \defrule{r-rec} &
            \Cut\x{\Rec[y]\x.P}{\Corec[y]\x.Q} & \red & \Cut\y{P}{Q}
            \\
            \defrule{r-choice} &
            \Choice{P_1}{P_2} & \red & P_i & i\in\set{1,2}
            \\
            \defrule{r-cut} &
            \Cut\x{P}{R} & \red & \Cut\x{Q}{R} & \text{if $P \red Q$}
            \\
            \defrule{r-struct} &
            P & \red & Q & \text{if $P \pcong R \red Q$}
        \end{array}
    \end{math}
\end{table}

The operational semantics of the calculus is given in terms of the
\emph{structural precongruence} relation $\pcong$ and the \emph{reduction
relation} $\red$ defined in \cref{tab:semantics}.
As usual, structural precongruence relates processes that are syntactically
different but semantically equivalent. In particular, \refrule{s-link} states
that linking $x$ with $y$ is the same as linking $y$ with $x$, whereas
\refrule{s-comm} and \refrule{s-assoc} state the expected  commutativity and
associativity laws for parallel composition. Concerning the latter, the side
condition $x\in\fn{Q}$ makes sure that $Q$ (the process brought closer to $P$
when the relation is read from left to right) is indeed connected with $P$ by
means of the channel $x$.
Note that \refrule{s-assoc} only states the right-to-left associativity of
parallel composition and that the left-to-right associativity law
$\Cut\x{\Cut\y{P}{Q}}{R} \pcong \Cut\y{P}{\Cut\x{Q}{R}}$ is derivable when
$x\in\fn{Q}$.
The reduction relation is mostly unremarkable. Links are reduced with
\refrule{r-link} by effectively merging the linked channels. All the reductions
that involve the interaction between processes except \refrule{r-unit} create
new continuations channels that connect the reducts. The rule \refrule{r-choice}
models the non-deterministic choice between two behaviors. Finally,
\refrule{r-cut} and \refrule{r-struct} close reductions by cuts and structural
precongruence.
In the following we write $\wred$ for the reflexive, transitive closure of
$\red$ and we say that $P$ is \emph{stuck} if there is no $Q$ such that $P \red
Q$.

We conclude this section by formalizing fair termination. To this aim, we
introduce the notion of \emph{run} as a maximal execution of a process.
Hereafter $\omega$ stands for the lowest transfinite ordinal number and $o$
ranges over the elements of $\omega + 1$.

\begin{definition}[run]
    \label{def:run}
    A \emph{run} of $P$ is a sequence $(P_i)_{i\in o}$ where $o\in\omega+1$ and
    $P_0 = P$ and $P_i \red P_{i+1}$ for all $i+1\in o$ and either the sequence
    is infinite or it ends with a stuck process.
\end{definition}

Hereafter we use $\Run$ to range over runs.
Using runs we can define a range of termination properties for processes. In
particular, $P$ is \emph{terminating} if all of its runs are finite and $P$ is
\emph{weakly terminating} if it has a finite run. Fair termination is a
termination property somewhat in between termination and weak termination in
which only the ``fair'' runs of a process are taken into account insofar its
termination is concerned. There exist several notions of fair run corresponding
to different fairness
assumptions~\cite{Francez86,Kwiatkowska89,GlabbeekHofner19}. The notion of fair
run used here is an instance of the \emph{fair reachability of predicates} by
Queille and Sifakis~\cite{QueilleSifakis83}.

\begin{definition}[fair termination]
    \label{def:fair-termination}
    A run is \emph{fair} if it contains finitely many weakly terminating
    processes. We say that $P$ is \emph{fairly terminating} if every fair run of
    $P$ is finite.
\end{definition}

To better understand our notion of fair run, it may be useful to think of the
cases in which a run is \emph{not} fair. An \emph{unfair} run is necessarily
infinite and describes the execution of a process that always has the chance to
terminate but systematically avoids doing so. Think of the system modeled in
\cref{ex:buyer-seller}: the (only) run in which the buyer adds items to the
shopping cart forever and never pays the seller is unfair; any other run of the
system is fair and finite. So, the system in \cref{ex:buyer-seller} is not
terminating (it admits an infinite execution) but it is fairly terminating (all
the fair executions are finite).

The following result characterizes fair termination without using fair runs.
Most importantly, it provides us with the key proof principle for the soundness
of our type system.

\begin{theorem}[name=proof principle for fair termination,label=thm:fair-termination,restate=thmfairtermination]
    $P$ is fairly terminating if and only if $P \wred Q$ implies that $Q$ is
    weakly terminating.
\end{theorem}
% \begin{proof}
%     See \cref{sec:fair-termination}.
% \end{proof}

\cref{thm:fair-termination} says that any reasonable type system that ensures
weak process termination also ensures fair process termination. By
``reasonable'' we mean a type system for which type preservation (also know as
\emph{subject reduction}) holds, which is usually the case. Indeed, if we
consider a well-typed process $P$ such that $P \wred Q$, we can deduce that $Q$
is also well typed. Now, the soundness of the type system guarantees that $Q$ is
weakly terminating. By applying \cref{thm:fair-termination} from right to left,
we conclude that every well-typed $P$ is fairly terminating.

\section{Formulas and Types}
\label{sec:types}

The types of \piLIN are built using the multiplicative additive fragment of
linear logic enriched with least and greatest fixed points. In this section we
specify the syntax of types along with all the auxiliary notions that are needed
to present the type system and prove its soundness.

The syntax of pre-formulas relies on an infinite set of \emph{propositional
variables} ranged over by $X$ and $Y$ and is defined by the grammar below:
\[
    \textbf{Pre-formula}
    \quad
    \FormulaF, \FormulaG
    ::=
    \Zero \mid
    \Top \mid
    \One \mid
    \Bot \mid
    \FormulaF \choice \FormulaG \mid
    \FormulaF \branch \FormulaG \mid
    \FormulaF \tfork \FormulaG \mid
    \FormulaF \tjoin \FormulaG \mid
    \tmu\X.\Formula \mid
    \tnu\X.\Formula \mid
    X
\]

As usual, $\tmu$ and $\tnu$ are the binders of propositional variables and the
notions of free and bound variables are defined accordingly. We assume that the
body of fixed points extends as much as possible to the right of a pre-formula,
so $\tmu\X.X \choice \One$ means $\tmu\X.(X \choice \One)$ and not $(\tmu\X.X)
\choice \One$. We write $\subst\Formula\X$ for the capture-avoiding substitution
of all free occurrences of $X$ with $\Formula$. We write $\dual\Formula$ for the
\emph{dual} of $\Formula$, which is the involution defined by the equations
\[
    \begin{array}{*{6}{@{}r@{~}c@{~}l@{\quad}}@{}}
        \dual\Zero & = & \Top &
        \dual{(\FormulaF \choice \FormulaG)} & = & \dual\FormulaF \branch \dual\FormulaG &
        \dual{(\tmu\X.\Formula)} & = & \tnu\X.\dual\Formula &
        \\
        \dual\One & = & \Bot &
        \dual{(\FormulaF \tfork \FormulaG)} & = & \dual\FormulaF \tjoin \dual\FormulaG &
        \dual\X & = & X
    \end{array}
\]

A \emph{formula} is a closed pre-formula. In the context of \piLIN, formulas
describe how linear channels are used. Positive formulas (those built with the
constants $\Zero$ and $\One$, the connectives $\choice$ and $\tfork$ and the
least fixed point) indicate output operations whereas negative formulas (the
remaining forms) indicate input operations. \REVISION{The formulas $\FormulaF
\choice \FormulaG$ and $\FormulaF \branch \FormulaG$ describe a linear channel
used for sending/receiving a tagged channel of type $\FormulaF$ or $\FormulaG$.
The tag (either $\LeftTag$ or $\RightTag$) distinguishes between the two
possibilities. The formulas} $\FormulaF \tfork \FormulaG$ and $\FormulaF \tjoin
\FormulaG$ describe a linear channel used for sending/receiving a pair of
channels of type $\FormulaF$ and $\FormulaG$; $\tmu\X.\Formula$ and
$\tnu\X.\Formula$ describe a linear channel used for sending/receiving a channel
of type $\Formula\subst{\tmu\X.\Formula}\X$ or
$\Formula\subst{\tnu\X.\Formula}\X$ respectively. The constants $\One$ and
$\Bot$ describe a linear channel used for sending/receiving the unit. Finally,
the constants $\Zero$ and $\Top$ respectively describe channels on which nothing
can be sent and from which nothing can be received.

\begin{example}
    \label{ex:formulas-buyer-seller}
    Looking at the structure of $\Buyer$ and $\Seller$ in
    \cref{ex:buyer-seller}, we can make an educated guess on the type of the
    channel $x$ they use. Concerning $x$, we see that it is used according to
    $\FormulaF \eqdef \tmu\X.X \choice \One$ in $\Buyer$ and according to
    $\FormulaG \eqdef \tnu\X.X \branch \Bot$ in $\Seller$. Note that $\FormulaF =
    \dual\FormulaG$, suggesting that $\Buyer$ and $\Seller$ may interact
    correctly when connected.
    \eoe
\end{example}

We write $\subf$ for the \emph{subformula ordering}, that is the least partial
order such that $\FormulaF \subf \FormulaG$ if $\FormulaF$ is a subformula of
$\FormulaG$. For example, consider $\FormulaF \eqdef \tmu\X.\tnu\Y.X \choice Y$
and $\FormulaG \eqdef \tnu\Y.\FormulaF \choice Y$. Then we have $\FormulaF \subf
\FormulaG$ and $\FormulaG \not\subf \FormulaF$. When $\Formulas$ is a set of
formulas, we write $\minf\Formulas$ for its $\subf$-minimum formula if it
is defined.
Occasionally we let $\star$ stand for an arbitrary binary connective $\choice$,
$\tfork$, $\branch$, or $\tjoin$ and $\sigma$ stand for an arbitrary fixed point
operator $\tmu$ or $\tnu$.

When two \piLIN processes interact on some channel $x$, they may exchange other
channels on which their interaction continues. We can think of these subsequent
interactions stemming from a shared channel $x$ as being part of the same
conversation (the rich literature on \emph{sessions}~\cite{Honda93,HuttelEtAl16}
builds on this idea~\cite{DardhaGiachinoSangiorgi17}). The soundness proof of
the type system is heavily based on the proof of the cut elimination property of
\muMALL, which relies on the ability to uniquely identify the types of the
channels that belong to the same conversation and to trace conversations within
typing derivations. Following the literature on
\muMALL~\cite{BaeldeDoumaneSaurin16,Doumane17,BaeldeEtAl22}, we annotate
formulas with addresses.
We assume an infinite set $\AddressSet$ of \emph{atomic addresses},
$\dual\AddressSet$ being the set of their duals such that $\AddressSet \cap
\dual\AddressSet = \emptyset$ and $\dual{\dual\AddressSet{}} = \AddressSet$. We
use $a$ and $b$ to range over elements of $\AddressSet \cup \dual\AddressSet$.
An \emph{address} is a string $aw$ where $w \in \set{i,l,r}^*$. The dual of an
address is defined as $\dual{(aw)} = \dual{a}w$.
We use $\addressA$ and $\addressB$ to range over addresses, we write $\prefix$
for the prefix relation on addresses and we say that $\addressA$ and $\addressB$
are \emph{disjoint} if $\addressA \not\prefix \addressB$ and $\addressB
\not\prefix \addressA$.

A \emph{type} is a formula $\Formula$ paired with an address $\address$ written
$\Formula_\address$. We use $S$ and $T$ to range over types and we extend to
types several operations defined on formulas: we use logical connectives to
compose types so that $\FormulaF_{\address l} \mathbin\star \FormulaG_{\address
r} \eqdef (\FormulaF \mathbin\star \FormulaG)_\address$ and
$\sigma\X.\Formula_{\address i} \eqdef (\sigma\X.\Formula)_\address$; the dual
of a type is obtained by dualizing both its formula and its address, that is
$\dual{(\Formula_\address)} \eqdef \dual\Formula_{\dual\address}$; type
substitution preserves the address in the type within which the substitution
occurs, but forgets the address of the type being substituted, that is
$\FormulaF_\addressA\subst{\FormulaG_\addressB}\X \eqdef
\FormulaF\subst\FormulaG\X_\addressA$.

We often omit the address of constants (which represent terminated
conversations) and we write $\strip{S}$ for the formula obtained by forgetting
the address of $S$. Finally, we write $\tred$ for the least reflexive relation
on types such that $S_1 \star S_2 \tred S_i$ and $\sigma\X.S \tred
S\subst{\sigma\X.S}\X$.

\begin{example}
    \label{ex:types-buyer-seller}
    Consider once again the formula $\FormulaF \eqdef \tmu\X.X \choice \One$ that
    describes the behavior of $\Buyer$ (\cref{ex:formulas-buyer-seller}) and let
    $a$ be an arbitrary atomic address. We have
    \[
        \Formula_a
        \tred (\Formula \choice \One)_{ai}
        \tred \Formula_{ail}
        \tred (\Formula \choice \One)_{aili}
        \tred \One_{ailir}
    \]
    where the fact that the types in this sequence all share a common non-empty
    prefix `$a$' indicates that they belong to the same conversation.
    Note how the symbols $i$, $l$ and $r$ composing an address indicate the step
    taken in the syntax tree of types for making a move in this sequence: $i$
    means ``inside'', when a fixed point operator is unfolded, whereas $l$ and
    $r$ mean ``left'' and ``right'', when the corresponding branch of a
    connective is selected.
    \eoe
\end{example}

\section{Type System}
\label{sec:type-system}

We now present the typing rules for \piLIN.
As usual we introduce typing contexts to track the type of the names occurring
free in a process. A \emph{typing context} is a finite map from names to types
written $x_1 : S_1, \dots, x_n : S_n$. We use $\ContextC$ and $\ContextD$ to
range over contexts, we write $\dom\Context$ for the domain of $\Context$ and
$\ContextC,\ContextD$ for the union of $\ContextC$ and $\ContextD$ when
$\dom\ContextC \cap \dom\ContextD = \emptyset$.
Typing judgments have the form $\qtp\Context{P}$. We say that $P$ is \emph{quasi
typed} in $\Context$ if the judgment $\qtp\Context{P}$ is coinductively
derivable using the rules shown in \cref{tab:typing-rules} and described below.
For the time being we say ``quasi typed'' and not ``well typed'' because some
infinite derivations using the rules in \cref{tab:typing-rules} are invalid.
Well-typed processes are quasi-typed processes whose typing derivation satisfies
some additional validity conditions that we detail in \cref{sec:soundness}.

\begin{table}
    \caption{\label{tab:typing-rules}Typing rules for \piLIN.}
    \begin{mathpar}
        \inferrule{\mathstrut}{
            \qtp{x : \Formula_\addressA, y : \dual\Formula_\addressB}{\Link\x\y}
        }
        ~\defrule\LinkRule
        \and
        \inferrule{
            \qtp{\ContextC, x : S}{P}
            \\
            \qtp{\ContextD, x : \dual{S}}{Q}
        }{
            \qtp{\ContextC, \ContextD}{\Cut\x{P}{Q}}
        }
        ~\defrule\CutRule
        \and
        \inferrule{\mathstrut}{
            \qtp{\Context, x : \Top}{\Fail\x}
        }
        ~\defrule[fail]\FailRule
        \and
        \inferrule{
            \qtp\Context{P}
        }{
            \qtp{\Context, x : \Bot}{\Wait\x.P}
        }
        ~\defrule[wait]\WaitRule
        \and
        \inferrule{\mathstrut}{
            \qtp{x : \One}{\Close\x}
        }
        ~\defrule[close]\CloseRule
        \and
        \inferrule{
            \qtp{\Context, y : S, z : T}{P}
        }{
            \qtp{\Context, x : S \tjoin T}{\Join[z]\x\y.P}
        }
        ~\defrule[join]\JoinRule
        \and
        \inferrule{
            \qtp{\ContextC, y : S}{P}
            \\
            \qtp{\ContextD, z : T}{Q}
        }{
            \qtp{\ContextC, \ContextD, x : S \tfork T}{\Fork[z]\x\y{P}{Q}}
        }
        ~\defrule[fork]\ForkRule
        \and
        \inferrule{
            \qtp{\Context, y : S}{P}
            \\
            \qtp{\Context, y : T}{Q}
        }{
            \qtp{\Context, x : S \branch T}{\Case[y]\x{P}{Q}}
        }
        ~\defrule[case]\CaseRule
        \and
        \inferrule{
            \qtp{\Context, y : S_i}{P}
        }{
            \qtp{\Context, x : S_1 \choice S_2}{\Select[y]{\InTag_i}\x.P}
        }
        ~\defrule[select]\SelectRule
        \and
        \inferrule{
            \qtp{\Context, y : S\subst{\tnu\X.S}X}{P}
        }{
            \qtp{\Context, x : \tnu\X.S}{\Corec[y]\x.P}
        }
        ~\defrule[corec]\CorecRule
        \and
        \inferrule{
            \qtp{\Context, y : S\subst{\tmu\X.S}X}{P}
        }{
            \qtp{\Context, x : \tmu\X.S}{\Rec[y]\x.P}
        }
        ~\defrule[rec]\RecRule
        \and
        \inferrule{
            \qtp\Context{P}
            \\
            \qtp\Context{Q}
        }{
            \qtp\Context{\Choice{P}{Q}}
        }
        ~\defrule\ChoiceRule
    \end{mathpar}
\end{table}

Rule \refrule{\LinkRule} states that a link $\Link\x\y$ is quasi typed provided
that $x$ and $y$ have dual types, but not necessarily dual addresses.
Rule \refrule{\CutRule} states that a process composition $\Cut\x{P}{Q}$ is
quasi typed provided that $P$ and $Q$ use the linear channel $x$ in
complementary ways, one according to some type $S$ and the other according to
the dual type $\dual{S}$. Note that the context $\ContextC,\ContextD$ in the
conclusion of the rule is defined provided that $\ContextC$ and $\ContextD$ have
disjoint domains. This condition entails that $P$ and $Q$ do not share any
channel other than $x$ ensuring that the interation between $P$ and $Q$ may
proceed without deadlocks.
%
% In line with the intution we have provided for the rank of a process, the rank
% of a composition is the sum of the ranks of the components.
%
Rule \refrule[fail]{\FailRule} deals with a process that receives an empty
message from channel $x$. Since this cannot happen, we allow the process to be
quasi typed in any context.
%
% A $\Fail\x$ process has null rank, since it performs no internal choices. There
% is no typing rule introducing $\Zero$, which would describe the behavior of a
% process that sends an empty message, as expected.
%
Rules \refrule[close]{\CloseRule} and \refrule[wait]{\WaitRule} concern the
exchange of units. The former rule states that $\Close\x$ is quasi typed in a
context that contains a single association for the $x$ channel with type $\One$,
whereas the latter rule removes $x$ from the context (hence from the set of
usable channels), requiring the continuation process to be quasi typed in the
remaining context.
%
% The rank of $\Close\x$ is zero since it performs no internal choices, whereas
% the rank of $\Wait\x.P$ is the same as that of $P$.
%
Rules \refrule[fork]{\ForkRule} and \refrule[join]{\JoinRule} concern the
exchange of pairs. The former rule requires the two forked processes $P$ and $Q$
to be quasi typed in the respective contexts enriched with associations for the
continuation channels $y$ and $z$ being created. The latter rule requires the
continuation process to be quasi typed in a context enriched with the channels
extracted from the received pair.
%
% As in \refrule{\CutRule}, the rank of a pair output is the sum of the ranks of
% the processes being forked, whereas the rank of a pair input simply
% corresponds to that of the continuation process.
%
Rules \refrule[select]{\SelectRule} and \refrule[case]{\CaseRule} deal with the
exchange of disjoint sums in the expected way.
%
% Their structure is quite unremarkable, except for the fact that the rank of a
% sum input process is the maximum among the ranks of its branches. In this way,
% we conservatively approximate the number of internal choices that the process
% has to perform before it terminates, not knowing which branch will actually be
% selected during its execution.
%
Rules \refrule[rec]{\RecRule} and \refrule[corec]{\CorecRule} deal with fixed
point operators by unfolding the (co)recursive type of the channel $x$. As in
\muMALL, the two rules have exactly the same structure despite the fact that the
two fixed point operators being used are dual to each other. Clearly, the
behavior of least and greatest fixed points must be distinguished by some other
means, as we will see in \cref{sec:soundness} when discussing the validity of a
typing derivation.
Finally, \refrule{\ChoiceRule} deals with non-deterministic choices by requiring
that each branch of a choice must be quasi typed in exactly the same typing
context as the conclusion.
%
% The rank of the non-deterministic choice is one plus that of the $k$-indexed
% branch that is elected as the branch that leads to termination. This is also
% the branch that is deterministically selected in the soundness proof of the
% type system for proving that a well-typed process weakly terminates.

Besides the structural constraints imposed by the typing rules, we implicitly
require that the types in the range of all typing contexts have pairwise
disjoint addresses. This condition ensures that it is possible to uniquely trace
a communication protocol in a typing derivation: if we have two channels $x$ and
$y$ associated with two types $\FormulaF_\addressA$ and $\FormulaG_\addressB$
such that $\addressA \prefix \addressB$, then we know that $y$ is a continuation
resulting from a communication that started from $x$. In a sense, $x$ and $y$
represent different moments in the same conversation.

% The typing rules in \cref{tab:typing-rules} except \refrule{\ChoiceRule} are in
% one-to-one correspondence with those of the \muMALL proof
% system~\cite{BaeldeDoumaneSaurin16,Doumane17}. Concerning \refrule{\ChoiceRule},
% it does not alter in any way the context of the process being typed and the fact
% that the rank in the conclusion is strictly larger than at least one of the
% ranks in the premises implies that the type system is a conservative extension
% of \muMALL. That is, if $\qtp[n]{x_1:S_1,\dots,x_n:S_n}{P}$ is coinductively
% derivable using the rules in \cref{tab:typing-rules}, then $\vdash
% S_1,\dots,S_n$ is coinductively derivable in \muMALL. The converse is also true,
% although the proof term $P$ is not uniquely determined.

\begin{example}[buyer and seller]
    \label{ex:typing-buyer-seller}
    Let us show that the system described in \cref{ex:buyer-seller} is quasi
    typed. To this aim, let $\FormulaF \eqdef \tmu\X.X \choice \One$ and
    $\FormulaG \eqdef \tnu\X.X \branch \Bot$ respectively be the formulas
    describing the behavior of $\Buyer$ and $\Seller$ on the channel $x$. Note
    that $\FormulaG = \dual\FormulaF$ and let $a$ be an arbitrary atomic
    address. We derive
    \[
        \begin{prooftree}
            \[
                \[
                    \[
                        \mathstrut\smash\vdots
                        \justifies
                        \qtp{
                            x : \FormulaF_{ail}
                        }{
                            \Call\Buyer\x
                        }
                    \]
                    \justifies
                    \qtp{
                        x : (\FormulaF \choice \One)_{ai}                        
                    }{
                        \Select\AddTag\x.\Call\Buyer\x
                    }
                    \using\refrule[select]\SelectRule
                \]
                \qquad
                \[
                    \[
                        \justifies
                        \qtp{
                            x : \One
                        }{
                            \Close\x
                        }
                        \using\refrule[close]{\CloseRule}
                    \]
                    \justifies
                    \qtp{
                        x : (\FormulaF \choice \One)_{ai}
                    }{
                        \Select\PayTag\x.\Close\x
                    }
                    \using\refrule[select]{\SelectRule}
                \]
                \justifies
                \qtp{
                    x : (\FormulaF \choice \One)_{ai}
                }{
                    \Choice{
                        \Select\AddTag\x.\Call\Buyer\x
                    }{
                        \Select\PayTag\x.\Close\x
                    }
                }
                \using\refrule{\ChoiceRule}
            \]
            \justifies
            \qtp{
                x : \FormulaF_a
            }{
                \Call\Buyer\x
            }
            \using\refrule[rec]{\RecRule}
        \end{prooftree}
    \]
    and also
    \[
        \begin{prooftree}
            \[
                \[
                    \mathstrut\smash\vdots
                    \justifies
                    \qtp{
                        x : \FormulaG_{\dual{a}il},
                        y : \One
                    }{
                        \Call\Seller{x,y}
                    }
                \]
                \qquad
                \[
                    \[
                        \justifies
                        \qtp{
                            y : \One
                        }{
                            \Close\y
                        }
                        \using\refrule[close]{\CloseRule}
                    \]
                    \justifies
                    \qtp{
                        x : \Bot,
                        y : \One
                    }{
                        \Wait\x.\Close\y
                    }
                    \using\refrule[wait]{\WaitRule}
                \]
                \justifies
                \qtp{
                    x : (\FormulaG \branch \Bot)_{\dual{a}i},
                    y : \One
                }{
                    \CaseX\x{
                        \Call\Seller{x,y}
                    }{
                        \Wait\x.\Close\y
                    }            
                }
                \using\refrule[case]{\CaseRule}
            \]
            \justifies
            \qtp{
                x : \FormulaG_{\dual{a}},
                y : \One
            }{
                \Call\Seller{x,y}
            }
            \using\refrule[corec]{\CorecRule}
        \end{prooftree}
    \]
    showing that $\Buyer$ and $\Seller$ are quasi typed. Note that both
    derivations are infinite, but for dual reasons. In $\Buyer$ the infinite
    branch corresponds to the behavior in which $\Buyer$ chooses to add one more
    item to the shopping cart. This choice is made independently of the behavior
    of other processes in the system. In $\Seller$, the infinite branch
    corresponds to the behavior in which $\Seller$ receives one more $\AddTag$
    message from $\Buyer$.
    By combining these derivations we obtain
    \[
        \begin{prooftree}
            \[
                \mathstrut\smash\vdots
                \justifies
                \qtp{
                    x : \FormulaF_a
                }{
                    \Call\Buyer\x
                }
            \]
            \qquad
            \[
                \mathstrut\smash\vdots
                \justifies
                \qtp{
                    x : \FormulaG_{\dual{a}},
                    y : \One
                }{
                    \Call\Seller{x,y}
                }
            \]
            \justifies
            \qtp{
                y : \One
            }{
                \Cut\x{\Call\Buyer\x}{\Call\Seller{x,y}}
            }
            \using\refrule{\CutRule}
        \end{prooftree}
    \]
    showing that the system as a whole is quasi typed.
    \eoe
\end{example}

As we have anticipated, there exist infinite typing derivations that are unsound
from a logical standpoint, because they allow us to prove $\Zero$ or the empty
sequent. For example, if we consider the non-terminating process $\Omega(x) =
\Choice{\Call\Omega\x}{\Call\Omega\x}$ we obtain the infinite derivation
\begin{equation}
    \label{eq:omega}
    \begin{prooftree}
        \[
            \mathstrut\smash\vdots
            \justifies
            \qtp{x : \Zero}{\Call\Omega\x}
        \]
        \qquad
        \[
            \mathstrut\smash\vdots
            \justifies
            \qtp{x : \Zero}{\Call\Omega\x}
        \]
        \justifies
        \qtp{x : \Zero}{\Call\Omega\x}
        \using\refrule\ChoiceRule
    \end{prooftree}
\end{equation}
showing that $\Call\Omega\x$ is quasi typed. As illustrated by the next example,
there exist non-terminating processes that are quasi typed also in logically
sound contexts.

\begin{example}[compulsive buyer]
    \label{ex:typing-compulsive-buyer}
    Consider the following variant of the $\Buyer$ process
    \[
        \Buyer(x,z) = \Rec\x.\Select\AddTag\x.\Call\Buyer{x,z}
    \] 
    that models a ``compulsive buyer'', namely a buyer that adds infinitely many
    items to the shopping cart but never pays. Using $\FormulaF \eqdef \tmu\X.X
    \choice \One$ and an arbitrary atomic address $a$ we can build the following
    infinite derivation
    \[
        \begin{prooftree}
            \[
                \[
                    \mathstrut\smash\vdots
                    \justifies
                    \qtp{
                        x : \Formula_{ail}
                    }{
                        \Call\Buyer{x}
                    }
                \]
                \justifies
                \qtp{
                    x : (\Formula \choice \One)_{ai}
                }{
                    \Select\AddTag\x.\Call\Buyer{x}
                }
                \using\refrule[select]{\SelectRule}
            \]
            \justifies
            \qtp{
                x : \Formula_a
            }{
                \Call\Buyer{x}
            }
            \using\refrule[rec]{\RecRule}
        \end{prooftree}
    \]
    showing that this process is quasi typed. By combining this derivation with
    the one for $\Seller$ in \cref{ex:typing-buyer-seller} we obtain a
    derivation establishing that $\Cut\x{\Call\Buyer\x}{\Call\Seller{x,y}}$ is
    quasi typed in the context $y : \One$, although this composition cannot
    terminate.
    \eoe
\end{example}

\section{From Quasi-Typed to Well-Typed Processes}
\label{sec:soundness}

To rule out unsound derivations like those in
\cref{eq:omega,ex:typing-compulsive-buyer} it is necessary to impose a validity
condition on derivations~\cite{BaeldeDoumaneSaurin16,Doumane17}. Roughly
speaking, \muMALL's validity condition requires every infinite branch of a
derivation to be supported by the continuous unfolding of a greatest fixed
point. In order to formalize this condition, we start by defining
\emph{threads}, which are sequences of types describing sequential interactions
at the type level.

\begin{definition}[thread]
    \label{def:thread}
    A \emph{thread} of $S$ is a sequence of types $(S_i)_{i\in o}$ for some
    $o\in\omega + 1$ such that $S_0 = S$ and $S_i \tred S_{i+1}$ whenever
    $i+1\in o$.
\end{definition}

Hereafter we use $t$ to range over threads. For example, if we consider
$\Formula \eqdef \tmu\X.X \choice \One$ from \cref{ex:buyer-seller} we have that
$t \eqdef (\Formula_a,(\Formula\choice\One)_{ai},\Formula_{ail},\dots)$ is an
infinite thread of $\Formula_a$. A thread is \emph{stationary} if it has an
infinite suffix of equal types. The above thread $t$ is not stationary.

Among all threads, we are interested in finding those in which a $\tnu$-formula
is unfolded infinitely often. These threads, called $\tnu$-threads, are precisely
defined thus:

\begin{definition}[$\tnu$-thread]
    \label{def:nu-thread}
    Let $t = (S_i)_{i\in\omega}$ be an infinite thread, let $\strip{t}$ be the
    corresponding sequence $(\strip{S_i})_{i\in\omega}$ of formulas and let
    $\InfOften{t}$ be the set of elements of $\strip{t}$ that occur infinitely
    often in $\strip{t}$. We say that $t$ is a \emph{$\tnu$-thread} if
    $\minf\InfOften{t}$ is \REVISION{defined and is} a $\tnu$-formula.
\end{definition}

If we consider the infinite thread $t$ above, we have $\InfOften{t} =
\set{\Formula, \Formula \choice \One}$ and $\minf\InfOften{t} = \Formula$, so $t$
is \emph{not} a $\tnu$-thread because $\Formula$ is not a $\tnu$-formula.
Consider instead $\FormulaF \eqdef \tnu\X.\tmu\Y.X \choice Y$ and $\FormulaG
\eqdef \tmu\Y.\FormulaF \choice Y$ and observe that $\FormulaG$ is the
``unfolding'' of $\FormulaF$. Now $t_1 \eqdef (\FormulaF_a, \FormulaG_{ai},
(\FormulaF \choice \FormulaG)_{aii}, \FormulaF_{aiil}, \dots)$ is a thread of
$\FormulaF_a$ such that $\InfOften{t_1} = \set{\FormulaF, \FormulaG, \FormulaF
\choice \FormulaG}$ and we have $\minf\InfOften{t_1} = \FormulaF$ because
$\FormulaF \subf \FormulaG$, so $t_1$ is a $\tnu$-thread.
If, on the other hand, we consider the thread $t_2 \eqdef (\FormulaF_a,
\FormulaG_{ai}, (\FormulaF \choice \FormulaG)_{aii}, \FormulaG_{aiir},
(\FormulaF \choice \FormulaG)_{aiiri}, \dots)$ such that $\InfOften{t_2} =
\set{\FormulaG, \FormulaF \choice \FormulaG}$ we have $\minf\InfOften{t_2} =
\FormulaG$ because $\FormulaG \subf \FormulaF \choice \FormulaG$, so $t_2$ is
not a $\tnu$-thread.
Intuitively, the $\subf$-minimum formula among those that occur infinitely often
in a thread is the outermost fixed point operator that is being unfolded
infinitely often. \REVISION{It is possible to show that this minimum formula is
always well defined~\cite{Doumane17}.} If such minimum formula is a greatest
fixed point operator, then the thread is a $\tnu$-thread.

Now we proceed by identifying threads along branches of typing derivations. To
this aim, we provide a precise definition of \emph{branch}.

\begin{definition}[branch]
    \label{def:branch}
    A \emph{branch} of a typing derivation is a sequence
    $(\qtp{\Context_i}{P_i})_{i\in o}$ of judgments for some $o\in\omega+1$ such
    that $\qtp{\Context_0}{P_0}$ occurs somewhere in the derivation and
    $\qtp{\Context_{i+1}}{P_{i+1}}$ is a premise of the rule application that
    derives $\qtp{\Context_i}{P_i}$ whenever $i+1\in o$.
\end{definition}

An infinite branch is valid if supported by a $\tnu$-thread that originates
somewhere therein.

\begin{definition}[valid branch]
    \label{def:valid-branch}
    Let $\gamma = (\qtp{\Context_i}{P_i})_{i\in\omega}$ be an infinite branch in
    a derivation. We say that $\gamma$ is \emph{valid} if there exists
    $j\in\omega$ such that $(S_k)_{k\geq j}$ is a non-stationary $\tnu$-thread
    and $S_k$ is in the range of $\Context_k$ for every $k \geq j$.
\end{definition}

For example, the infinite branch in the typing derivation for $\Seller$ of
\cref{ex:buyer-seller} is valid since it is supported by the $\tnu$-thread
$(\FormulaG_{\dual{a}}, (\FormulaG \branch \Bot)_{\dual{a}i},
\FormulaG_{\dual{a}il},\dots)$ where $\FormulaG \eqdef \tnu\X.X \branch \Bot$
happens to be the $\subf$-minimum formula that is unfolded infinitely often.
On the other hand, the infinite branch in the typing derivation for $\Buyer$ of
\cref{ex:typing-compulsive-buyer} is invalid, because the only infinite thread
in it is $(\FormulaF_a, (\FormulaF \choice \One)_{ai}, \FormulaF_{ail}, \dots)$
which is not a $\tnu$-thread.

A \muMALL derivation is valid if so is every infinite branch in
it~\cite{BaeldeDoumaneSaurin16,Doumane17}. \REVISION{For the purpose of ensuring
fair termination}, this condition is too strong because some infinite branches
in a typing derivation may correspond to unfair executions that, by definition,
we neglect insofar its termination is concerned. For example, the infinite
branch in the derivation for $\Buyer$ of \cref{ex:typing-buyer-seller}
corresponds to an unfair run in which the buyer insists on adding items to the
shopping cart, despite it periodically has a chance of paying the seller and
terminate the interaction. That typing derivation for $\Buyer$ would be
considered an invalid proof in \muMALL because the infinite branch is not
supported by a $\tnu$-thread (in fact, there is a $\tmu$-formula that is
unfolded infinitely many times along that branch, as in
\cref{ex:typing-compulsive-buyer}).

It is generally difficult to understand if a branch corresponds to a fair or
unfair run because the branch describes the evolution of an incomplete process
whose behavior is affected by the interactions it has with processes found in
other branches of the derivation.
However, we can detect (some) unfair branches by looking at the
non-deterministic choices they traverse, since choices are made autonomously by
processes. To this aim, we introduce the notion of \emph{rank} to estimate the
least number of choices a process can possibly make during its lifetime.

% If there are no choices, the typing derivation corresponds to a \muMALL proof
% and the termination of the process is a consequence of cut elimination.

% the same fairness criterion adopted for process runs
% (\cref{def:fair-termination}) to discriminate between fair and unfair branches.
% Nonetheless, we can \emph{rank} processes with an element of the set $\RankSet
% \eqdef \Nat\cup\set\infty$ estimating the number of choices they have to perform
% in order to terminate. The smaller the rank, the closer the process is to
% termination. We use $\rankR$ and $\rankS$ to range over ranks and assume that
% $\RankSet$ is equipped  with the expected total order $\leq$ and operation $+$
% such that $\rank+\infty = \infty+\rank = \infty$.

\begin{definition}[rank]
    \label{def:rank}
    Let $\rankR$ and $\rankS$ range over the elements of $\RankSet \eqdef
    \Nat\cup\set\infty$ equipped with the expected total order $\leq$ and
    operation $+$ such that $\rank+\infty = \infty+\rank = \infty$.
    The \emph{rank} of a process $P$, written $\rankof{P}$, is the least element
    of $\RankSet$ such that
    \[
        \begin{array}{@{}r@{~}c@{~}l@{}}
            \rankof{\Link\x\y} & = & 0 \\
            \rankof{\Fail\x} & = & 0 \\
            \rankof{\Close\x} & = & 0 \\
            \rankof{\Wait\x.P} & = & \rankof{P} \\
        \end{array}
        \qquad
        \begin{array}{@{}r@{~}c@{~}l@{}}
            \rankof{\Join[z]\x\y.P} & = & \rankof{P} \\
            \rankof{\Select[y]{\InTag_i}\x.P} & = & \rankof{P} \\
            \rankof{\Rec[y]\x.P} & = & \rankof{P} \\
            \rankof{\Corec[y]\x.P} & = & \rankof{P} \\
        \end{array}
        \qquad
        \begin{array}{@{}r@{~}c@{~}l@{}}
            \rankof{\Case[y]\x{P}{Q}} & = & \max\set{\rankof{P},\rankof{Q}} \\
            \rankof{\Choice{P}{Q}} & = & 1 + \min\set{\rankof{P},\rankof{Q}} \\
            \rankof{\Cut\x{P}{Q}} & = & \rankof{P} + \rankof{Q} \\
            \rankof{\Fork[z]\x\y{P}{Q}} & = & \rankof{P} + \rankof{Q} \\
        \end{array}
    \]
\end{definition}

Roughly, the rank of terminated processes is $0$, that of processes with a
single continuation $P$ coincides with the rank of $P$, and that of processes
spawning two continuations $P$ and $Q$ is the sum of the ranks of $P$ and $Q$.
Then, the rank of a sum input with continuations $P$ and $Q$ is conservatively
estimated as the maximum of the ranks of $P$ and $Q$, since we do not know which
one will be taken, whereas the rank of a choice with continuations $P$ and $Q$
is 1 plus the minimum of the ranks of $P$ and $Q$.
If we take $\Buyer$ and $\Seller$ from \cref{ex:buyer-seller} we have
$\rankof{\Call\Buyer\x} = 1$ and $\rankof{\Call\Seller{x,y}} = 0$. We also have
$\rankof{\Call\Omega\x} = \infty$.
Note that $\rankof{P}$ only depends on the structure of $P$ but not on the
actual names occurring in $P$, so it is well and uniquely defined as the least
solution of a finite system of equations (see \cref{sec:extra-rank}).

\begin{definition}
    \label{def:fair-branch}
    A branch is \emph{fair} if it traverses finitely many, finitely-ranked
    choices.
\end{definition}

A finitely-ranked choice is at finite distance from a region of the process in
which there are no more choices. An \emph{unfair} branch gets close to such
region infinitely often, but systematically avoids entering it.
Note that every finite branch is also fair, but there are fair branches that are
infinite. For instance, all the infinite branches of the derivation in
\cref{eq:omega} and the only infinite branch in the derivation for
$\Call\Seller{x,y}$ of \cref{ex:typing-buyer-seller} are fair since they do not
traverse any finitely-ranked choice. On the contrary, the only infinite branch
in the derivation for $\Call\Buyer\x$ of the \cref{ex:typing-buyer-seller} is
unfair since it traverses infinitely many finitely-ranked choices. All fair
branches in the same derivation for $\Buyer$ are finite.

At last we can define our notion of well-typed process.

\begin{definition}[well-typed process]
    \label{def:wtp}
    We say that $P$ is \emph{well typed} in $\Context$, written
    $\wtp\Context{P}$, if the judgment $\qtp\Context{P}$ is derivable and each
    fair, infinite branch in its derivation is valid.
\end{definition}

Note that $\Omega$ is ill typed since the fair, infinite branches in
\cref{eq:omega} are all invalid. We can now formulate the key properties of
well-typed processes, starting from subject reduction.

\begin{theorem}[name=subject reduction,label=thm:subject-reduction,restate=thmsubjectreduction]
    If $\wtp\Context{P}$ and $P \red Q$ then $\wtp\Context{Q}$.
\end{theorem}
% \begin{proof}
%     See \cref{sec:subject-reduction}.
% \end{proof}

\REVISION{
All reductions in \cref{tab:semantics} except those for non-deterministic
choices correspond to cut-elimination steps in a quasi typing derivation. As
an illustration, below is a fragment of derivation tree for two processes
exchanging a pair of $y$ and $z$ on channel $x$.
\[
    \begin{prooftree}
        \[
            \[
                \mathstrut\smash\vdots
                \justifies
                \qtp{\ContextC, y : S}{P}
            \]
            \[
                \mathstrut\smash\vdots
                \justifies
                \qtp{\ContextD, z : T}{Q}
            \]
            \justifies
            \qtp{
                \ContextC, \ContextD, x : S \tfork T
            }{
                \Fork[z]\x\y{P}{Q}
            }
            \using\refrule[fork]\ForkRule
        \]
        \[
            \[
                \mathstrut\smash\vdots
                \justifies
                \qtp{\ContextC', y : \dual{S}, z : \dual{T}}{R}
            \]
            \justifies
            \qtp{\ContextC', x : \dual{S} \tjoin \dual{T}}{\Join[z]\x\y.R}
            \using\refrule[join]\JoinRule
        \]
        \justifies
        \qtp{
            \ContextC, \ContextD, \ContextC'
        }{
            \Cut\x{
                \Fork[z]\x\y{P}{Q}
            }{
                \Join[z]\x\y.R
            }
        }
        \using\refrule\CutRule
    \end{prooftree}
\]

As the process reduces, the quasi typing derivation is rearranged so that
the cut on $x$ is replaced by two cuts on $y$ and $z$. The resulting quasi
typing derivation is shown below.
\[
    \begin{prooftree}
        \[
            \mathstrut\smash\vdots
            \justifies
            \qtp{\ContextC, y : S}{P}
        \]
        \[
            \[
                \mathstrut\smash\vdots
                \justifies
                \qtp{\ContextD, z : T}{Q}
            \]
            \[
                \mathstrut\smash\vdots
                \justifies
                \qtp{\ContextC', y : \dual{S}, z : \dual{T}}{R}
            \]
            \justifies
            \qtp{
                \ContextD, \ContextC', y : \dual{S}
            }{
                \Cut\z{Q}{R}
            }
            \using\refrule\CutRule
        \]
        \justifies
        \qtp{
            \ContextC, \ContextD, \ContextC'
        }{
            \Cut\y{P}{\Cut\z{Q}{R}}
        }
        \using\refrule\CutRule            
    \end{prooftree}
\]
}

It is also interesting to observe that, when $P \red Q$, the reduct $Q$ is well
typed in the same context as $P$ but its rank may be different. In particular,
the rank of $Q$ can be \emph{greater} than the rank of $P$. Recalling that the
rank of a process estimates the number of choices that the process must perform
to terminate, the fact that the rank of $Q$ increases means that $Q$ \emph{moves
away} from termination instead of getting closer to it (we will see an instance
where this phenomenon occurs in \cref{ex:parallel-programming}). What really
matters is that a well-typed process is weakly terminating. This is the second
key property ensured by our type system.

\begin{lemma}[name=weak termination,label=lem:weak-termination,restate=lemweaktermination]
    If $\wtp{x : \One}{P}$ then $P \wred \Close\x$.
\end{lemma}
% \begin{proof}
%     See \cref{sec:extra-type-system}.
% \end{proof}

The proof of \cref{lem:weak-termination} is a refinement of the cut elimination
property of \muMALL. Essentially, the only new case we have to handle is when a
choice $\Choice{P_1}{P_2}$ ``emerges'' towards the bottom of the typing
derivation, meaning that it is no longer guarded by any action. In this case, we
reduce the choice to the $P_i$ with smaller rank, which is guaranteed to lay on
a fair branch of the derivation.
An auxiliary result used in the proof of \cref{lem:weak-termination} is that our type system is a conservative extension of \muMALL.

\begin{lemma}[name={},restate=lemconservativity]
    If $\wtp{x_1:S_1,\dots,x_n:S_n}{P}$ then $\vdash S_1,\dots,S_n$ is derivable
    in \muMALL.
\end{lemma}

The property that well-typed processes can always successfully terminate is a
simple consequence of \cref{thm:subject-reduction,lem:weak-termination}.

\begin{theorem}[name=soundness,label=thm:soundness,restate=thmsoundness]
    If $\wtp{x : \One}{P}$ and $P \wred Q$ then $Q \wred \Close\x$.
\end{theorem}

\cref{thm:soundness} entails all the good properties we expect from well-typed
processes: \emph{failure freedom} (no unguarded sub-process $\Fail\y$ ever
appears), \emph{deadlock freedom} (if the process stops it is terminated),
\emph{lock freedom}~\cite{Kobayashi02,Padovani14} (every pending action can be
completed in finite time) and \emph{junk freedom} (every channel can be
depleted).
The combination of \cref{thm:fair-termination,thm:soundness} also guarantees the
termination of every fair run of the process.

\begin{corollary}[name=fair termination,label=cor:fair-termination,restate=corfairtermination]
    If $\wtp{x : \One}{P}$ then $P$ is fairly terminating.
\end{corollary}

Observe that zero-ranked process do not contain any non-deterministic choice. In
that case, every infinite branch in their typing derivation is fair and our
validity condition coincides with that of \muMALL. As a consequence, we obtain
the following strengthening of \cref{cor:fair-termination}:\LP{Controllare se
\`e gi\`a dimostrato con un altro nome, se no aggiungere la prova?}

\begin{proposition}
    If $\wtp{x : \One}{P}$ and $\rankof{P} = 0$ then $P$ is terminating.
\end{proposition}

For regular processes (those consisting of finitely many distinct sub-trees, up
to renaming of bound names) it is possible to easily adapt the algorithm that
decides the validity of a \muMALL proof so that it decides the validity of a
\piLIN typing derivation. The algorithm is sketched in more detail in
\cref{sec:proof-decidability}.

\begin{example}[parallel programming]
    \label{ex:parallel-programming}
    \newcommand{\Work}{\textit{Work}}
    \newcommand{\Gather}{\textit{Gather}}
    \newcommand{\ComplexTag}{\mktag{complex}}
    \newcommand{\SimpleTag}{\mktag{simple}}
    In this example we see a \piLIN modeling of a \emph{parallel programming
    pattern} whereby a $\Work$ process creates an unbounded number of workers
    each one dedicated to an independent task and a $\Gather$ process collects
    and combines the partial results from the workers. The processes $\Work$ and
    $\Gather$ are defined as follows:
    \begin{align*}
        \Work(x) & =
        \Rec\x.\parens{
            \Choice{
                \Select\ComplexTag\x.\Fork\x\y{\Close\y}{\Call\Work\x}
            }{
                \Select\SimpleTag\x.\Close\x
            }
        }
        \\
        \Gather(x,z) & =
        \Corec\x.
        \Case\x{
            \Join\x\y.\Wait\y.\Call\Gather{x,z}
        }{
            \Wait\x.\Close\z
        }
    \end{align*}

    At each iteration, the $\Work$ process non-deterministically decides whether
    the task is $\ComplexTag$ (left hand side of the choice) or $\SimpleTag$
    (right hand side of the choice). In the first case, it bifurcates into a new
    worker, which in the example simply sends a unit on $y$, and another
    instance of itself. In the second case it terminates.
    The $\Gather$ process joins the results from all the workers before
    signalling its own termination by sending a unit on $z$. Note that the
    number of actions $\Gather$ has to perform before terminating is unbounded,
    as it depends on the non-deterministic choices made by $\Work$.

    Below is a typing derivation for $\Work$ where $\FormulaF \eqdef
    \tmu\X.(\One \tfork X) \choice \One$ and $a$ is an arbitrary atomic address:
    \[
        \begin{prooftree}
            \[
                \[
                    \[
                        \[
                            \justifies
                            \qtp{
                                y : \One
                            }{
                                \Close\y
                            }
                            \using\refrule[close]\CloseRule
                        \]
                        \quad
                        \[
                            \mathstrut\smash\vdots
                            \justifies
                            \qtp{
                                x : \FormulaF_{ailr}
                            }{
                                \Call\Work\x
                            }
                        \]
                        \justifies
                        \qtp{
                            x : (\One \tfork \FormulaF)_{ail}
                        }{
                            \Fork\x\y{\Close\y}{\Call\Work\x}
                        }
                        \using\refrule[fork]\ForkRule
                    \]
                    \justifies
                    \qtp{
                        x : ((\One \tfork \FormulaF) \choice \One)_{ai}
                    }{
                        \Select\ComplexTag\x\dots
                    }
                    \using\refrule[select]\SelectRule
                \]
                \[
                    \justifies
                    \qtp{
                        x : ((\One \tfork \FormulaF) \choice \One)_{ai}
                    }{
                        \Select\SimpleTag\x.\Close\x
                    }
                    \using\refrule[close]\CloseRule,\refrule[select]\SelectRule
                \]
                \justifies
                \qtp{
                    x : ((\One \tfork \FormulaF) \choice \One)_{ai}
                }{
                    \Choice{
                        \Select\ComplexTag\x.\Fork\x\y{\Close\y}{\Call\Work\x}
                    }{
                        \Select\SimpleTag\x.\Close\x
                    }
                }
                \using\refrule\ChoiceRule
            \]
            \justifies
            \qtp{
                x : \FormulaF_\address
            }{
                \Call\Work\x
            }
            \using\refrule[rec]\RecRule
        \end{prooftree}
    \]

    Note that the only infinite branch in this derivation is unfair because it
    traverses infinitely many choices with rank $1 = \rankof{\Call\Work\x}$. So,
    $\Work$ is well typed.

    Concerning $\Gather$, we obtain the following typing derivation where
    $\FormulaG \eqdef \tnu\X.(\Bot \tjoin X) \branch \Bot$:
    \[
        \begin{prooftree}
            \[
                \[
                    \[
                        \[
                            \mathstrut\smash\vdots
                            \justifies
                            \qtp{
                                x : \FormulaG_{\dual ailr},
                                z : \One
                            }{
                                \Call\Gather{x,z}    
                            }
                        \]
                        \justifies
                        \qtp{
                            x : \FormulaG_{\dual ailr},
                            y : \Bot,
                            z : \One
                        }{
                            \Wait\y.\Call\Gather{x,z}
                        }
                        \using\refrule[wait]\WaitRule
                    \]
                    \justifies
                    \qtp{
                        x : (\Bot \tjoin \FormulaG)_{\dual ail},
                        z : \One
                    }{
                        \Join\x\y.\Wait\y.\Call\Gather{x,z}
                    }
                    \using\refrule[join]\JoinRule
                \]
                \[
                    \justifies
                    \qtp{
                        x : \Bot,
                        z : \One
                    }{
                        \Wait\x.\Close\z
                    }
                    \using\refrule[close]\CloseRule,\refrule[wait]\WaitRule
                \]
                \justifies
                \qtp{
                    x : ((\Bot \tjoin \FormulaG) \branch \Bot)_{\dual ai},
                    z : \One
                }{
                    \Case\x{
                        \Join\x\y.\Wait\y.\Call\Gather{x,z}
                    }{
                        \Wait\x.\Close\z
                    }            
                }
                \using\refrule[case]\CaseRule
            \]
            \justifies
            \qtp{
                x : \FormulaG_{\dual a},
                z : \One
            }{
                \Call\Gather{x,z}
            }
            \using\refrule[corec]\CorecRule
        \end{prooftree}
    \]

    Here too there is just one infinite branch, which is fair and supported by
    the $\tnu$-thread $t = (\FormulaG_{\dual a}, ((\Bot \tjoin \FormulaG) \branch
    \Bot)_{\dual ai}, (\Bot \tjoin \FormulaG)_{\dual ail}, \FormulaG_{\dual
    ailr}, \dots)$. Indeed, all the formulas in $\strip{t}$ occur infinitely
    often and $\minf\InfOften{t} = \FormulaG$ which is a $\tnu$-formula.
    Hence, $\Gather$ is well typed and so is the composition
    $\Cut\x{\Call\Work\x}{\Call\Gather{x,z}}$ in the context $z : \One$. We
    conclude that the program is fairly terminating, despite the fact that the
    composition of $\Work$ and $\Gather$ may grow arbitrarily large because
    $\Work$ may spawn an unbounded number of workers.
    \eoe
\end{example}

\begin{example}[forwarder]
    \label{ex:forwarder}
    \newcommand{\Forwarder}{\textit{Fwd}}
    In this example we illustrate a deterministic, well-typed process that
    unfolds a least fixed point infinitely many times. In particular, we
    consider once again the formulas $\FormulaF \eqdef \tmu\X.X \choice \One$ and
    $\FormulaG \eqdef \tnu\X.X \branch \Bot$ and the process $\Forwarder$ defined
    by the equation
    \[
        \Forwarder(x,y) =
        \Corec\x.
        \Rec\y.
        \Case\x{
            \Left\y.
            \Call\Forwarder{x,y}
        }{
            \Right\y.
            \Wait\x.
            \Close\y
        }
    \]
    which forwards the sequence of messages received from channel $x$ to channel
    $y$. We derive
    \[
        \begin{prooftree}
            \[
                \[
                    \[
                        \mathstrut\smash\vdots
                        \justifies
                        \qtp{
                            x : \FormulaG_{ail},
                            y : \FormulaF_{bil}
                        }{
                            \Call\Forwarder{x,y}
                        }
                    \]
                    \justifies
                    \qtp{
                        x : \FormulaG_{ail},
                        y : (\FormulaF \choice \One)_{bi}
                    }{
                        \Left\y.\Call\Forwarder{x,y}
                    }
                    \using\refrule[select]{\SelectRule}
                \]                
                \[
                    \[
                        \justifies
                        \qtp{
                            x : \Bot,
                            y : \One
                        }{
                            \Wait\x.\Close\y
                        }
                        \using\refrule[close]\CloseRule,\refrule[wait]\WaitRule
                    \]
                    \justifies
                    \qtp{
                        x : \Bot,
                        y : (\FormulaF \choice \One)_{bi}
                    }{
                        \Right\y.\Wait\x.\Close\y
                    }
                    \using\refrule[select]{\SelectRule}
                \]
                \justifies
                \qtp{
                    x : (\FormulaG \branch \Bot)_{ai},
                    y : (\FormulaF \choice \One)_{bi}
                }{
                    \Case\x{
                        \Left\y.
                        \Call\Forwarder{x,y}
                    }{
                        \Right\y.
                        \Wait\x.
                        \Close\y
                    }            
                }
                \using\refrule[case]{\CaseRule}
            \]
            \justifies
            \qtp{x : \FormulaG_a, y : \FormulaF_b}{\Call\Forwarder{x,y}}
            \using\refrule[corec]{\CorecRule}, \refrule[rec]{\RecRule}
        \end{prooftree}
    \]
    and observe that $\rankof{\Call\Forwarder{x,y}} = 0$. This typing derivation
    is valid because the only infinite branch is fair and supported by the
    $\tnu$-thread of $\FormulaG_a$.
    Note that $\FormulaF = \dual\FormulaG$ and that the derivation proves an
    instance of \refrule{\LinkRule}. In general, the axiom is admissible in
    \muMALL~\cite{BaeldeDoumaneSaurin16}.
    \eoe
\end{example}

\newcommand{\Player}{\textit{Player}}
\newcommand{\Machine}{\textit{Machine}}

\begin{example}[slot machine]
    \label{ex:slot-machine}
    Rank finiteness is not a necessary condition for well typedness. As an
    example, consider the system $\Cut\x{\Call\Player\x}{\Call\Machine{x,y}}$
    where
    \[
        \begin{array}{@{}r@{~}c@{~}l@{}}
            \Player(x) & = &
            \Rec\x.
            \parens{
                \Choice{
                    \Select\PlayTag\x.
                    \Case\x{
                        \Call\Player\x
                    }{
                        \Rec\x.\Select\QuitTag\x.\Close\x
                    }
                }{
                    \Select\QuitTag\x.\Close\x
                }
            }
            \\
            \Machine(x,y) & = &
            \Corec\x.
            \Case\x{
                \Choice{
                    \Select\WinTag\x.\Call\Machine{x,y}
                }{
                    \Select\LoseTag\x.\Call\Machine{x,y}
                }
            }{
                \Wait\x.\Close\y
            }
        \end{array}
    \]
    which models a game between a player and a slot machine. At each round, the
    player decides whether to $\PlayTag$ or to $\QuitTag$. In the first case,
    the slot machine answers with either $\WinTag$ or $\LoseTag$. If the player
    $\WinTag$s, it also $\QuitTag$s. Otherwise, it repeats the same behavior.
    It is possible to show that $\qtp{x : \FormulaF_a}{\Call\Player\x}$ and
    $\qtp{x : \FormulaG_{\dual a}, y : \One}{\Call\Machine{x,y}}$ are derivable
    where $\FormulaF \eqdef \tmu\X.(X \branch X) \choice \One$ and $\FormulaG
    \eqdef \tnu\X.(X \choice X) \branch \Bot$, see \cref{sec:extra-soundness} for
    the details.
    The only infinite branch in the derivation for $\Player$ is unfair since
    $\rankof{\Call\Player\x} = 1$, so $\Player$ is well typed.
    There are infinitely many branches in the derivation for $\Machine$
    accounting for all the sequences of $\WinTag$ and $\LoseTag$ choices that
    can be made. Since $\rankof{\Call\Machine{x,y}} = \infty$, all these
    branches are fair but also valid.
    So, the system as a whole is well typed.
    \eoe
\end{example}

\section{Related Work}
\label{sec:related}

On account of the known encodings of sessions into the linear
$\pi$-calculus~\cite{Kobayashi02b,DardhaGiachinoSangiorgi17,ScalasDardhaHuYoshida17},
\piLIN belongs to the family of process calculi providing logical foundations to
sessions and session types. Some representatives of this family are
\piDILL~\cite{CairesPfenningToninho16} and its variant equipped with a circular
proof theory \cite{DerakhshanPfenning19,Derakhshan21}, \CP~\cite{Wadler14} and
\muCP~\cite{LindleyMorris16}, among others.
There are two main aspects distinguishing \piLIN from these calculi. The first
one is that these calculi take sessions as a native feature. This fact can be
appreciated both at the level of processes, where session endpoints are
\emph{linearized} resources that can be used \emph{multiple times} albeit in a
sequential way, and also at the level of types, where the interpretation of the
$\FormulaF \tfork \FormulaG$ and $\FormulaF \tjoin \FormulaG$ formulas is skewed
so as to distinguish the type $\FormulaF$ of the message being sent/received on
a channel from the type $\FormulaG$ of the channel after the exchange has taken
place.
In contrast, \piLIN adopts a more fundamental communication model based on
\emph{linear} channels, and is thus closer to the spirit of the encoding of
linear logic proofs into the $\pi$-calculus proposed by Bellin and
Scott~\cite{BellinScott94} while retaining the same expressiveness of the
aforementioned calculi. To some extent, \piLIN is also more general than those
calculi, since a formula $\FormulaF \tfork \FormulaG$ may be interpreted as a
protocol that bifurcates into two independent sub-protocols $\FormulaF$ and
$\FormulaG$ (we have seen an instance of this pattern in
\cref{ex:parallel-programming}). So, \piLIN is natively equipped with the
capability of modeling some multiparty interactions, in addition to all of the
binary ones.
A session-based communication model identical to \piLIN, but whose type system
is based on intuitionistic rather than classical linear logic, has been
presented by DeYoung et al.~\cite{DeYoungCairesPfenningToninho12}. In that work,
the authors advocate the use of explicit continuations with the purpose of
modeling an asynchronous communication semantics and they prove the equivalence
between such model and a buffered session semantics. However, they do not draw a
connection between the proposed calculus and the linear
$\pi$-calculus~\cite{KobayashiPierceTurner99} through the encoding of binary
sessions~\cite{Kobayashi02b,DardhaGiachinoSangiorgi17} and, in the type system,
the multiplicative connectives are still interpreted asymmetrically.
The second aspect that distinguishes \piLIN from the other calculi is its type
system, which is still deeply rooted into linear logic and yet it ensures fair
termination instead of
progress~\cite{CairesPfenningToninho16,Wadler14,QianKavvosBirkedal21},
termination~\cite{LindleyMorris16} or strong
progress~\cite{DerakhshanPfenning19,Derakhshan21}. Fair termination entails
progress, strong progress and lock freedom~\cite{Kobayashi02,Padovani14}, but at
the same time it does not always rule out processes admitting infinite
executions. Simply, infinite executions are deemed unrealistic because they are
unfair.

Another difference between \piLIN and other calculi based on linear logic is
that its operational semantics is completely ordinary, in the sense that it does
not include commuting conversions, reductions under prefixes, or the swapping of
communication prefixes. The cut elimination result of \muMALL, on which the
proof of \cref{thm:soundness} is based, works by reducing cuts from the bottom
of the derivation instead of from the
top~\cite{BaeldeDoumaneSaurin16,Doumane17,BaeldeEtAl22}. As a consequence, it is
not necessary to reduce cuts guarded by prefixes or to push cuts deep into the
derivation tree to enable key reductions in \piLIN processes.

Some previous works \cite{CicconePadovani22,CicconeDagninoPadovani22} have
studied type systems ensuring the fair termination of binary and multiparty
sessions. Although the enforced property is the same and the communication
models of these works are closely related to the one we consider here, the used
techniques are quite different and the families of well-typed processes induced
by the two type systems are not comparable.
One aspect that is shared among all of these works, including the present one,
is the use of a \emph{rank} annotation or function that estimates how far a
process is from termination. In previous works for session-based
communications~\cite{CicconePadovani22,CicconeDagninoPadovani22}, ranks account
for the number of sessions that processes create and the number of times they
use subtyping before they terminate. Since ranks are required to be finite, this
means that deterministic processes can only create a finite number of sessions
and can only use subtyping a finite number of times. As a consequence, the
forwarder in \cref{ex:forwarder} is ill typed according to those typing
disciplines because it applies subtyping (in an implicit way, whenever it sends
a tag) an unbounded number of times.
At the same time, a ``compulsive'' variant of the player in
\cref{ex:slot-machine} that keeps playing until it wins is well typed in
previous type systems~\cite{CicconePadovani22,CicconeDagninoPadovani22} but ill
typed in the present one. This difference stems, at least in part, from the fact
that the previous type systems support processes defined using general
recursion, whereas \piLIN's type system relies on the duality between recursion
and corecursion. A deeper understanding of these differences requires further
investigation.

The extension of calculi based on linear logic with non-deterministic features
has recently received quite a lot of attention. Rocha and
Caires~\cite{RochaCaires21} have proposed a session calculus with shared cells
and non-deterministic choices that can model mutable state. Their typing rule
for non-deterministic choices is the same as our own, but in their calculus
choices do not reduce. Rather, they keep track of every possible evolution of a
process to be able to prove a confluence result.
Qian et al.~\cite{QianKavvosBirkedal21} introduce \emph{coexponentials}, a new
pair of modalities that enable the modeling of concurrent clients that compete
in order to interact with a shared server that processes requests sequentially.
In this setting, non-determinism arises from the unpredictable order in which
different clients are served. Interestingly, the coexponentials are derived by
resorting to their semantics in terms of least and greatest fixed points. For
this reason, the cut elimination result of \muMALL might be useful to attack the
termination proof in their setting.

\section{Concluding Remarks}
\label{sec:conclusion}

We have studied a conservative extension of
\muMALL~\cite{BaeldeDoumaneSaurin16,Doumane17,BaeldeEtAl22}, an infinitary proof
system of multiplicative additive linear logic with fixed points, that serves as
type system for \piLIN, a linear $\pi$-calculus with (co)recursive types.
Well-typed processes fairly terminate.

One drawback of the proposed type system is that establishing whether a
quasi-typed process is also well typed requires a global check on the whole
typing derivation. \REVISION{The need for a global check seems to arise commonly
in infinitary proof
systems~\cite{DaxHofmannLange06,Doumane17,BaeldeDoumaneSaurin16,BaeldeEtAl22},
so} an obvious aspect to investigate is whether the analysis can be localized. A
possible source of inspiration for devising a local type system for \piLIN might
come from the work of Derakhshan and Pfenning~\cite{DerakhshanPfenning19}. They
propose a compositional technique for dealing with infinitary typing derivations
in a session calculus, although their type system is limited to the additive
fragment of linear logic.

Fair subtyping~\cite{Padovani13,Padovani16} is a refinement of the standard
subtyping relation for session types~\cite{GayHole05} that preserves fair
termination and that plays a key role in our previous type system ensuring fair
termination for binary sessions~\cite{CicconePadovani22}. Given the rich
literature exploring the connections between linear logic and session types,
\piLIN and its type system might provide the right framework for investigating
the logical foundations of fair subtyping.

\bibliography{main}

\begin{thebibliography}{10}

\bibitem{Abramsky94}
Samson Abramsky.
\newblock Proofs as processes.
\newblock {\em Theor. Comput. Sci.}, 135(1):5--9, 1994.
\newblock \href {https://doi.org/10.1016/0304-3975(94)00103-0}
  {\path{doi:10.1016/0304-3975(94)00103-0}}.

\bibitem{AptFrancezKatz87}
Krzysztof~R. Apt, Nissim Francez, and Shmuel Katz.
\newblock Appraising fairness in languages for distributed programming.
\newblock In {\em Conference Record of the Fourteenth Annual {ACM} Symposium on
  Principles of Programming Languages, Munich, Germany, January 21-23, 1987},
  pages 189--198. {ACM} Press, 1987.
\newblock \href {https://doi.org/10.1145/41625.41642}
  {\path{doi:10.1145/41625.41642}}.

\bibitem{BaeldeEtAl22}
David Baelde, Amina Doumane, Denis Kuperberg, and Alexis Saurin.
\newblock Bouncing threads for circular and non-wellfounded proofs.
\newblock In {\em Proceedings of the 37th Annual ACM/IEEE Symposium on Logic in
  Computer Science, LICS, 2022}, 2022.
\newblock URL: \url{https://arxiv.org/abs/2005.08257}.

\bibitem{BaeldeDoumaneSaurin16}
David Baelde, Amina Doumane, and Alexis Saurin.
\newblock Infinitary proof theory: the multiplicative additive case.
\newblock In Jean{-}Marc Talbot and Laurent Regnier, editors, {\em 25th {EACSL}
  Annual Conference on Computer Science Logic, {CSL} 2016, August 29 -
  September 1, 2016, Marseille, France}, volume~62 of {\em LIPIcs}, pages
  42:1--42:17. Schloss Dagstuhl - Leibniz-Zentrum f{\"{u}}r Informatik, 2016.
\newblock \href {https://doi.org/10.4230/LIPIcs.CSL.2016.42}
  {\path{doi:10.4230/LIPIcs.CSL.2016.42}}.

\bibitem{BellinScott94}
Gianluigi Bellin and Philip~J. Scott.
\newblock On the pi-calculus and linear logic.
\newblock {\em Theor. Comput. Sci.}, 135(1):11--65, 1994.
\newblock \href {https://doi.org/10.1016/0304-3975(94)00104-9}
  {\path{doi:10.1016/0304-3975(94)00104-9}}.

\bibitem{CairesPerez16}
Lu{\'{\i}}s Caires and Jorge~A. P{\'{e}}rez.
\newblock Multiparty session types within a canonical binary theory, and
  beyond.
\newblock In Elvira Albert and Ivan Lanese, editors, {\em Formal Techniques for
  Distributed Objects, Components, and Systems - 36th {IFIP} {WG} 6.1
  International Conference, {FORTE} 2016, Held as Part of the 11th
  International Federated Conference on Distributed Computing Techniques,
  DisCoTec 2016, Heraklion, Crete, Greece, June 6-9, 2016, Proceedings}, volume
  9688 of {\em Lecture Notes in Computer Science}, pages 74--95. Springer,
  2016.
\newblock \href {https://doi.org/10.1007/978-3-319-39570-8\_6}
  {\path{doi:10.1007/978-3-319-39570-8\_6}}.

\bibitem{CairesPfenningToninho16}
Lu{\'{\i}}s Caires, Frank Pfenning, and Bernardo Toninho.
\newblock Linear logic propositions as session types.
\newblock {\em Math. Struct. Comput. Sci.}, 26(3):367--423, 2016.
\newblock \href {https://doi.org/10.1017/S0960129514000218}
  {\path{doi:10.1017/S0960129514000218}}.

\bibitem{CicconeDagninoPadovani22}
Luca Ciccone, Francesco Dagnino, and Luca Padovani.
\newblock Fair termination of multiparty sessions.
\newblock In Karim Ali and Jan Vitek, editors, {\em 36th European Conference on
  Object-Oriented Programming, {ECOOP} 2022, June 6-10, 2022, Berlin, Germany},
  volume 222 of {\em LIPIcs}, pages 26:1--26:26. Schloss Dagstuhl -
  Leibniz-Zentrum f{\"{u}}r Informatik, 2022.
\newblock \href {https://doi.org/10.4230/LIPIcs.ECOOP.2022.26}
  {\path{doi:10.4230/LIPIcs.ECOOP.2022.26}}.

\bibitem{CicconePadovani20}
Luca Ciccone and Luca Padovani.
\newblock A dependently typed linear {\(\pi\)}-calculus in agda.
\newblock In {\em {PPDP} '20: 22nd International Symposium on Principles and
  Practice of Declarative Programming, Bologna, Italy, 9-10 September, 2020},
  pages 8:1--8:14. {ACM}, 2020.
\newblock \href {https://doi.org/10.1145/3414080.3414109}
  {\path{doi:10.1145/3414080.3414109}}.

\bibitem{CicconePadovani22}
Luca Ciccone and Luca Padovani.
\newblock Fair termination of binary sessions.
\newblock {\em Proc. {ACM} Program. Lang.}, 6({POPL}):1--30, 2022.
\newblock \href {https://doi.org/10.1145/3498666} {\path{doi:10.1145/3498666}}.

\bibitem{DardhaGiachinoSangiorgi17}
Ornela Dardha, Elena Giachino, and Davide Sangiorgi.
\newblock Session types revisited.
\newblock {\em Inf. Comput.}, 256:253--286, 2017.
\newblock \href {https://doi.org/10.1016/j.ic.2017.06.002}
  {\path{doi:10.1016/j.ic.2017.06.002}}.

\bibitem{DaxHofmannLange06}
Christian Dax, Martin Hofmann, and Martin Lange.
\newblock A proof system for the linear time {\(\mathrm{\mu}\)}-calculus.
\newblock In S.~Arun{-}Kumar and Naveen Garg, editors, {\em {FSTTCS} 2006:
  Foundations of Software Technology and Theoretical Computer Science, 26th
  International Conference, Kolkata, India, December 13-15, 2006, Proceedings},
  volume 4337 of {\em Lecture Notes in Computer Science}, pages 273--284.
  Springer, 2006.
\newblock \href {https://doi.org/10.1007/11944836\_26}
  {\path{doi:10.1007/11944836\_26}}.

\bibitem{Derakhshan21}
Farzaneh Derakhshan.
\newblock {\em Session-Typed Recursive Processes and Circular Proofs}.
\newblock PhD thesis, Carnegie Mellon University, 2021.
\newblock URL:
  \url{https://www.andrew.cmu.edu/user/fderakhs/publications/Dissertation_Farzaneh.pdf}.

\bibitem{DerakhshanPfenning19}
Farzaneh Derakhshan and Frank Pfenning.
\newblock Circular proofs as session-typed processes: {A} local validity
  condition.
\newblock {\em CoRR}, abs/1908.01909, 2019.
\newblock URL: \url{http://arxiv.org/abs/1908.01909}, \href
  {http://arxiv.org/abs/1908.01909} {\path{arXiv:1908.01909}}.

\bibitem{DeYoungCairesPfenningToninho12}
Henry DeYoung, Lu{\'{\i}}s Caires, Frank Pfenning, and Bernardo Toninho.
\newblock Cut reduction in linear logic as asynchronous session-typed
  communication.
\newblock In Patrick C{\'{e}}gielski and Arnaud Durand, editors, {\em Computer
  Science Logic (CSL'12) - 26th International Workshop/21st Annual Conference
  of the EACSL, {CSL} 2012, September 3-6, 2012, Fontainebleau, France},
  volume~16 of {\em LIPIcs}, pages 228--242. Schloss Dagstuhl - Leibniz-Zentrum
  f{\"{u}}r Informatik, 2012.
\newblock \href {https://doi.org/10.4230/LIPIcs.CSL.2012.228}
  {\path{doi:10.4230/LIPIcs.CSL.2012.228}}.

\bibitem{Doumane17}
Amina Doumane.
\newblock {\em On the infinitary proof theory of logics with fixed points.
  (Th{\'{e}}orie de la d{\'{e}}monstration infinitaire pour les logiques
  {\`{a}} points fixes)}.
\newblock PhD thesis, Paris Diderot University, France, 2017.
\newblock URL: \url{https://tel.archives-ouvertes.fr/tel-01676953}.

\bibitem{Francez86}
Nissim Francez.
\newblock {\em Fairness}.
\newblock Texts and Monographs in Computer Science. Springer, 1986.
\newblock \href {https://doi.org/10.1007/978-1-4612-4886-6}
  {\path{doi:10.1007/978-1-4612-4886-6}}.

\bibitem{GardnerLaneveWischik07}
Philippa Gardner, Cosimo Laneve, and Lucian Wischik.
\newblock Linear forwarders.
\newblock {\em Inf. Comput.}, 205(10):1526--1550, 2007.
\newblock \href {https://doi.org/10.1016/j.ic.2007.01.006}
  {\path{doi:10.1016/j.ic.2007.01.006}}.

\bibitem{GayHole05}
Simon~J. Gay and Malcolm Hole.
\newblock Subtyping for session types in the pi calculus.
\newblock {\em Acta Informatica}, 42(2-3):191--225, 2005.
\newblock \href {https://doi.org/10.1007/s00236-005-0177-z}
  {\path{doi:10.1007/s00236-005-0177-z}}.

\bibitem{GayVasconcelos10}
Simon~J. Gay and Vasco~Thudichum Vasconcelos.
\newblock Linear type theory for asynchronous session types.
\newblock {\em J. Funct. Program.}, 20(1):19--50, 2010.
\newblock \href {https://doi.org/10.1017/S0956796809990268}
  {\path{doi:10.1017/S0956796809990268}}.

\bibitem{GrumbergFrancezKatz84}
Orna Grumberg, Nissim Francez, and Shmuel Katz.
\newblock Fair termination of communicating processes.
\newblock In {\em Proceedings of the Third Annual ACM Symposium on Principles
  of Distributed Computing}, PODC '84, pages 254--265, New York, NY, USA, 1984.
  Association for Computing Machinery.
\newblock \href {https://doi.org/10.1145/800222.806752}
  {\path{doi:10.1145/800222.806752}}.

\bibitem{Honda93}
Kohei Honda.
\newblock Types for dyadic interaction.
\newblock In Eike Best, editor, {\em {CONCUR} '93, 4th International Conference
  on Concurrency Theory, Hildesheim, Germany, August 23-26, 1993, Proceedings},
  volume 715 of {\em Lecture Notes in Computer Science}, pages 509--523.
  Springer, 1993.
\newblock \href {https://doi.org/10.1007/3-540-57208-2\_35}
  {\path{doi:10.1007/3-540-57208-2\_35}}.

\bibitem{HondaVasconcelosKubo98}
Kohei Honda, Vasco~Thudichum Vasconcelos, and Makoto Kubo.
\newblock Language primitives and type discipline for structured
  communication-based programming.
\newblock In Chris Hankin, editor, {\em Programming Languages and Systems -
  ESOP'98, 7th European Symposium on Programming, Held as Part of the European
  Joint Conferences on the Theory and Practice of Software, ETAPS'98, Lisbon,
  Portugal, March 28 - April 4, 1998, Proceedings}, volume 1381 of {\em Lecture
  Notes in Computer Science}, pages 122--138. Springer, 1998.
\newblock \href {https://doi.org/10.1007/BFb0053567}
  {\path{doi:10.1007/BFb0053567}}.

\bibitem{HondaYoshidaCarbone16}
Kohei Honda, Nobuko Yoshida, and Marco Carbone.
\newblock Multiparty asynchronous session types.
\newblock {\em J. {ACM}}, 63(1):9:1--9:67, 2016.
\newblock \href {https://doi.org/10.1145/2827695} {\path{doi:10.1145/2827695}}.

\bibitem{HuttelEtAl16}
Hans H{\"{u}}ttel, Ivan Lanese, Vasco~T. Vasconcelos, Lu{\'{\i}}s Caires, Marco
  Carbone, Pierre{-}Malo Deni{\'{e}}lou, Dimitris Mostrous, Luca Padovani,
  Ant{\'{o}}nio Ravara, Emilio Tuosto, Hugo~Torres Vieira, and Gianluigi
  Zavattaro.
\newblock Foundations of session types and behavioural contracts.
\newblock {\em {ACM} Comput. Surv.}, 49(1):3:1--3:36, 2016.
\newblock \href {https://doi.org/10.1145/2873052} {\path{doi:10.1145/2873052}}.

\bibitem{Kobayashi02}
Naoki Kobayashi.
\newblock A type system for lock-free processes.
\newblock {\em Inf. Comput.}, 177(2):122--159, 2002.
\newblock \href {https://doi.org/10.1006/inco.2002.3171}
  {\path{doi:10.1006/inco.2002.3171}}.

\bibitem{Kobayashi02b}
Naoki Kobayashi.
\newblock Type systems for concurrent programs.
\newblock In {\em 10th Anniversary Colloquium of UNU/IIST}, LNCS 2757, pages
  439--453. Springer, 2002.
\newblock Extended version at
  \url{http://www.kb.ecei.tohoku.ac.jp/~koba/papers/tutorial-type-extended.pdf}.
\newblock \href
  {https://doi.org/http://dx.doi.org/10.1007/978-3-540-40007-3_26}
  {\path{doi:http://dx.doi.org/10.1007/978-3-540-40007-3_26}}.

\bibitem{KobayashiPierceTurner99}
Naoki Kobayashi, Benjamin~C. Pierce, and David~N. Turner.
\newblock Linearity and the pi-calculus.
\newblock {\em {ACM} Trans. Program. Lang. Syst.}, 21(5):914--947, 1999.
\newblock \href {https://doi.org/10.1145/330249.330251}
  {\path{doi:10.1145/330249.330251}}.

\bibitem{Kwiatkowska89}
M.Z. Kwiatkowska.
\newblock Survey of fairness notions.
\newblock {\em Information and Software Technology}, 31(7):371--386, 1989.
\newblock \href {https://doi.org/10.1016/0950-5849(89)90159-6}
  {\path{doi:10.1016/0950-5849(89)90159-6}}.

\bibitem{Lamport00}
Leslie Lamport.
\newblock Fairness and hyperfairness.
\newblock {\em Distributed Comput.}, 13(4):239--245, 2000.
\newblock \href {https://doi.org/10.1007/PL00008921}
  {\path{doi:10.1007/PL00008921}}.

\bibitem{LindleyMorris16}
Sam Lindley and J.~Garrett Morris.
\newblock Talking bananas: structural recursion for session types.
\newblock In Jacques Garrigue, Gabriele Keller, and Eijiro Sumii, editors, {\em
  Proceedings of the 21st {ACM} {SIGPLAN} International Conference on
  Functional Programming, {ICFP} 2016, Nara, Japan, September 18-22, 2016},
  pages 434--447. {ACM}, 2016.
\newblock \href {https://doi.org/10.1145/2951913.2951921}
  {\path{doi:10.1145/2951913.2951921}}.

\bibitem{Padovani13}
Luca Padovani.
\newblock Fair subtyping for open session types.
\newblock In Fedor~V. Fomin, Rusins Freivalds, Marta~Z. Kwiatkowska, and David
  Peleg, editors, {\em Automata, Languages, and Programming - 40th
  International Colloquium, {ICALP} 2013, Riga, Latvia, July 8-12, 2013,
  Proceedings, Part {II}}, volume 7966 of {\em Lecture Notes in Computer
  Science}, pages 373--384. Springer, 2013.
\newblock \href {https://doi.org/10.1007/978-3-642-39212-2\_34}
  {\path{doi:10.1007/978-3-642-39212-2\_34}}.

\bibitem{Padovani14}
Luca Padovani.
\newblock Deadlock and lock freedom in the linear {\(\pi\)}-calculus.
\newblock In Thomas~A. Henzinger and Dale Miller, editors, {\em Joint Meeting
  of the Twenty-Third {EACSL} Annual Conference on Computer Science Logic
  {(CSL)} and the Twenty-Ninth Annual {ACM/IEEE} Symposium on Logic in Computer
  Science (LICS), {CSL-LICS} '14, Vienna, Austria, July 14 - 18, 2014}, pages
  72:1--72:10. {ACM}, 2014.
\newblock \href {https://doi.org/10.1145/2603088.2603116}
  {\path{doi:10.1145/2603088.2603116}}.

\bibitem{Padovani16}
Luca Padovani.
\newblock Fair subtyping for multi-party session types.
\newblock {\em Math. Struct. Comput. Sci.}, 26(3):424--464, 2016.
\newblock \href {https://doi.org/10.1017/S096012951400022X}
  {\path{doi:10.1017/S096012951400022X}}.

\bibitem{Padovani17}
Luca Padovani.
\newblock A simple library implementation of binary sessions.
\newblock {\em J. Funct. Program.}, 27:e4, 2017.
\newblock \href {https://doi.org/10.1017/S0956796816000289}
  {\path{doi:10.1017/S0956796816000289}}.

\bibitem{Padovani19}
Luca Padovani.
\newblock Context-free session type inference.
\newblock {\em {ACM} Trans. Program. Lang. Syst.}, 41(2):9:1--9:37, 2019.
\newblock \href {https://doi.org/10.1145/3229062} {\path{doi:10.1145/3229062}}.

\bibitem{QianKavvosBirkedal21}
Zesen Qian, G.~A. Kavvos, and Lars Birkedal.
\newblock Client-server sessions in linear logic.
\newblock {\em Proc. {ACM} Program. Lang.}, 5({ICFP}):1--31, 2021.
\newblock \href {https://doi.org/10.1145/3473567} {\path{doi:10.1145/3473567}}.

\bibitem{QueilleSifakis83}
Jean{-}Pierre Queille and Joseph Sifakis.
\newblock Fairness and related properties in transition systems - {A} temporal
  logic to deal with fairness.
\newblock {\em Acta Informatica}, 19:195--220, 1983.
\newblock \href {https://doi.org/10.1007/BF00265555}
  {\path{doi:10.1007/BF00265555}}.

\bibitem{RochaCaires21}
Pedro Rocha and Lu{\'{\i}}s Caires.
\newblock Propositions-as-types and shared state.
\newblock {\em Proc. {ACM} Program. Lang.}, 5({ICFP}):1--30, 2021.
\newblock \href {https://doi.org/10.1145/3473584} {\path{doi:10.1145/3473584}}.

\bibitem{ScalasDardhaHuYoshida17}
Alceste Scalas, Ornela Dardha, Raymond Hu, and Nobuko Yoshida.
\newblock A linear decomposition of multiparty sessions for safe distributed
  programming (artifact).
\newblock {\em Dagstuhl Artifacts Ser.}, 3(2):03:1--03:2, 2017.
\newblock \href {https://doi.org/10.4230/DARTS.3.2.3}
  {\path{doi:10.4230/DARTS.3.2.3}}.

\bibitem{ThiemannVasconcelos16}
Peter Thiemann and Vasco~T. Vasconcelos.
\newblock Context-free session types.
\newblock In Jacques Garrigue, Gabriele Keller, and Eijiro Sumii, editors, {\em
  Proceedings of the 21st {ACM} {SIGPLAN} International Conference on
  Functional Programming, {ICFP} 2016, Nara, Japan, September 18-22, 2016},
  pages 462--475. {ACM}, 2016.
\newblock \href {https://doi.org/10.1145/2951913.2951926}
  {\path{doi:10.1145/2951913.2951926}}.

\bibitem{GlabbeekHofner19}
Rob van Glabbeek and Peter H{\"{o}}fner.
\newblock Progress, justness, and fairness.
\newblock {\em {ACM} Comput. Surv.}, 52(4):69:1--69:38, 2019.
\newblock \href {https://doi.org/10.1145/3329125} {\path{doi:10.1145/3329125}}.

\bibitem{Wadler14}
Philip Wadler.
\newblock Propositions as sessions.
\newblock {\em J. Funct. Program.}, 24(2-3):384--418, 2014.
\newblock \href {https://doi.org/10.1017/S095679681400001X}
  {\path{doi:10.1017/S095679681400001X}}.

\end{thebibliography}

\appendix

\section{Properties of fair termination}
\label{sec:proof-language}

A key property of any notion of fair run is that of
\emph{feasibility}~\cite{AptFrancezKatz87} also known as \emph{machine
closure}~\cite{Lamport00}, stating that every finite reduction may be extended
to a fair run. It is not difficult to prove that this property is enjoyed by our
notion of fair run.

\begin{lemma}[name=feasibility,label=lem:feasibility]
    Let $P \wred Q$ and $(P_i)_{0\leq i\leq n}$ be the sequence of processes
    visited by these reductions. Then there exists a run $\Run$ of $Q$ such that
    $(P_i)_{0\leq i<n}\Run$ is a fair run of $P$.
\end{lemma}
\begin{proof}
    Let $(P_i)_{0\leq i\leq n}$ be the sequence of processes visited by the
    reductions $P \wred Q$, where $P_0 = P$ and $P_n = Q$. We distinguish two
    possibilities. If $Q$ is weakly terminating, then there exists a finite run
    $\Run$ of $Q$. Then, $(P_i)_{0\leq i<n}\Run$ is a finite fair run of $P$. If
    $Q$ is not weakly terminating, then $Q$ has at least one infinite run $\Run$
    in which no process is weakly terminating. Then $(P_i)_{0\leq i<n}\Run$ is
    an infinite fair run of $P$.
\end{proof}

\thmfairtermination*
\begin{proof}
    ($\Rightarrow$)
    Let $P_0\cdots P_n$ be the sequence of processes visited by the reductions
    $P \wred Q$. By \cref{lem:feasibility} we deduce that there exists a run
    $\Run$ of $Q$ such that $P_0\cdots P_n\Run$ is a fair run $\Run$ of $P$.
    From the hypothesis that $P$ is fairly terminating we deduce that $\Run$ is
    finite. That is, $Q$ is weakly terminating.

    ($\Leftarrow$)
    Let $\Run$ be a fair run of $P$ and suppose that it is infinite. Then there
    exists some $Q$ in this run such that $P \wred Q$ and $Q$ is not weakly
    terminating, which contradicts the hypothesis.
\end{proof}

\section{Proof of Theorem~\ref{thm:subject-reduction}}
\label{sec:subject-reduction}

\begin{lemma}[substitution]
\label{lem:substitution}
	If $\qtp{\Context, x : S}{P}$ and $y \not\in \dom\Context$, then
	$\qtp{\Context, y : S}{P\subst\y\x}$.
\end{lemma}
\begin{proof}
	A simple application of the coinduction principle.
\end{proof}

%%%%%%%%%%%%%%%%%%%%%%%%%
%%% Subject Conguence %%%
%%%%%%%%%%%%%%%%%%%%%%%%%

\begin{lemma}[quasi subject congruence]
\label{lem:quasi-subject-congruence}
	If\/ $\qtp\Context{P}$ and $P \pcong Q$, then $\qtp\Context{Q}$.
\end{lemma}
\begin{proof}
	By induction on the derivation of $P \pcong Q$ and by cases on the last rule
	applied. We only discuss the base cases.

	\begin{description}
		\item[Case \refrule{s-link}.]
		Then $P = \Link\x\y \pcong \Link\y\x = Q$.
		From \refrule{\LinkRule} we deduce that there exist $\Formula$,
		$\addressA$ and $\addressB$ such that $\Context = x :
		\Formula_\addressA, y : \dual{\Formula_\addressB}$. We conclude with an
		application of \refrule{\LinkRule}.

		\item[Case \refrule{s-comm}.]
		Then $P = \Cut\x{P_1}{P_2} \pcong \Cut\x{P_2}{P_1} = Q$. From
		\refrule{\CutRule} we deduce that there exist $\Context_1$,
		$\Context_2$, and $S$ such that
		\begin{itemize}
		\item $\Context = \Context_1, \Context_2$
		\item $\qtp{\Context_1, x : S}{P_1}$
		\item $\qtp{\Context_2, x : \dual{S}}{P_2}$
		\end{itemize}
		We conclude with an application of \refrule{\CutRule}.

		\item[Case \refrule{s-assoc}.]
		Then $P = \Cut\x{P_1}{\Cut\y{P_2}{P_3}} \pcong
		\Cut\y{\Cut\x{P_1}{P_2}}{P_3} = Q$ and $\x \in \fn{P_2}$.
		From \refrule{\CutRule} we deduce that there exist $\Context_1$,
		$\Context_2$, $\Context_3$, $S$, $T$ such that
		\begin{itemize}
		\item $\Context = \Context_1, \Context_2, \Context_3$
		\item $\qtp{\Context_1, x : S}{P_1}$
		\item $\qtp{\Context_2, x : \dual{S}, y : T}{P_2}$
		\item $\qtp{\Context_3, y : \dual{T}}{P_3}$
		\end{itemize}
		We conclude with two applications of \refrule{\CutRule}.
		\qedhere
	\end{description}
\end{proof}

%%%%%%%%%%%%%%%%%%%%%%%%%
%%% Subject Reduction %%%
%%%%%%%%%%%%%%%%%%%%%%%%%

\begin{lemma}[quasi subject reduction]
	\label{lem:quasi-subject-reduction}
	If $P \red Q$ then $\qtp\Context{P}$ implies $\qtp\Context{Q}$.
\end{lemma}
\begin{proof}
	By induction on $P \red Q$ and by cases on the last rule applied.

	\begin{description}
		\item[Case \refrule{r-link}.]
		Then $P = \Cut\x{\Link\x\y}{R} \red R\subst\y\x = Q$.
		From \refrule{\CutRule} and \refrule{\LinkRule} we deduce that there
		exist $\Formula_\addressA$, $\Formula_\addressB$ and $\ContextD$ such
		that $\Context = y : \dual\Formula_{\dual\addressB}, \ContextD$ and
		$\qtp{\ContextD, x : \dual\Formula_{\dual\addressB}}{R}$
		We conclude by applying \cref{lem:substitution}.

		\item[Case \refrule{r-unit}.]
		Then $P = \Cut\x{\Close\x}{\Wait{x}.Q} \red Q$.
		From \refrule\CutRule, \refrule[close]{\CloseRule} and
		\refrule[wait]{\WaitRule} we conclude $\qtp\Context{Q}$.

		\item[Case \refrule{r-pair}.]
		Then $P = \Cut\x{\Fork[z]\x\y{P_1}{P_2}}{\Join[z]\x\y.P_3} \red
		\Cut\y{P_1}{\Cut\z{P_2}{P_3}} = Q$.
		From \refrule{\CutRule}, \refrule[fork]{\ForkRule} and
		\refrule[join]{\JoinRule} we deduce that there exist $\Context_1$,
		$\Context_2$, $\Context_3$, $S$ and $T$ such that
		\begin{itemize}
		\item $\Context = \Context_1,\Context_2,\Context_3$
		\item $\qtp{\Context_1, y : S}{P_1}$
		\item $\qtp{\Context_2, z : T}{P_2}$
		\item $\qtp{\Context_3, y : \dual{S}, z : \dual{T}}{P_3}$
		\end{itemize}
		We conclude with two applications of \refrule{\CutRule}.
		
		\item[Case \refrule{r-sum}.]
		Then $P = \Cut\x{\Select[y]{\InTag_i}\x.R}{\Case[y]\x{P_1}{P_2}} \red
		\Cut\y{R}{P_i} = Q$ for some $i\in\set{1,2}$.
		From \refrule{\CutRule}, \refrule[select]{\SelectRule} and
		\refrule[case]{\CaseRule} we deduce that there exist $\Context_1$,
		$\Context_2$, $S_1$ and $S_2$ such that
		\begin{itemize}
		\item $\Context = \Context_1, \ContextD$
		\item $\qtp{\Context_1, y : S_i}{R}$
		\item $\qtp{\Context_2, y : \dual{S_1}}{P_1}$
		\item $\qtp{\Context_2, y : \dual{S_2}}{P_2}$
		\end{itemize}
		We conclude with an application of \refrule{\CutRule}.

		\item[Case \refrule{r-rec}.]
		Then $P = \Cut\x{\Rec[y]\x.P_1}{\Corec[y]\x.P_2} \red \Cut\y{P_1}{P_2} = Q$.
		From \refrule{\CutRule}, \refrule[rec]{\RecRule} and
		\refrule[corec]{\CorecRule} we deduce that there exist $\Context_1$,
		$\Context_2$ and $S$ such that
		\begin{itemize}
		\item $\Context = \Context_1, \Context_2$
		\item $\qtp{\Context_1, y : S\subst{\tmu\X.S}\X}{P_1}$
		\item $\qtp{\Context_2, y : \dual{S}\subst{\tnu\X.\dual{S}}{\X}}{P_2}$
		\end{itemize}
		We conclude with an application of \refrule{\CutRule}.

		\item[Case \refrule{r-choice}.]
		Then $P = \Choice{P_1}{P_2} \red P_i = Q$ for some $i \in \set{1,2}$.
		From \refrule{\ChoiceRule} we conclude that $\qtp\Context{P_i}$.

		\item[Case \refrule{r-cut}.]
		Then $P = \Cut\x{P'}{R} \red \Cut\x{Q'}{R} = Q$ and $P' \red Q'$.
		From \refrule{\CutRule} we deduce that there exist $\Context_1$,
		$\Context_2$ and $S$ such that
		\begin{itemize}
		\item $\Context = \Context_1, \Context_2$
		\item $\qtp{\Context_1, x : S}{P'}$
		\item $\qtp{\Context_2, x : \dual{S}}{R}$
		\end{itemize}
		Using the induction hypothesis on $P' \red Q'$ we deduce
		$\qtp{\Context_1, x : S}{Q'}$.
		We conclude with an application of \refrule{\CutRule}.

		\item[Case \refrule{r-struct}.]
		Then $P \pcong R \red Q$ for some $R$.
		From \cref{lem:quasi-subject-congruence} we deduce $\qtp\Context{R}$.
		Using the induction hypothesis on $R \red Q$ we conclude
		$\qtp\Context{Q}$.
		\qedhere
	\end{description}
\end{proof}

\thmsubjectreduction*
\begin{proof}
	From the hypothesis $\wtp\Context{P}$ we deduce $\qtp\Context{P}$. Using
	\cref{lem:quasi-subject-reduction} we deduce $\qtp\Context{Q}$.
	Now, let $\gamma$ be an infinite fair branch in the typing derivation for
	$\qtp\Context{Q}$. By inspecting the proof of
	\cref{lem:quasi-subject-reduction} we observe that $\gamma$ can be
	decomposed as $\gamma = \gamma_1\gamma_2$ where $\gamma_2$ is a branch in
	the typing derivation for $\qtp\Context{Q}$. From the fact that $\gamma$ is
	fair we deduce that so is $\gamma_2$.
	From the hypothesis $\wtp\Context{P}$ we deduce that $\gamma_2$ is valid,
	namely there is a $\tnu$-thread $t$ in it. Then $t$ is a $\tnu$-thread also of
	$\gamma$, that is $\gamma$ is also valid.
	We conclude $\wtp\Context{Q}$.
\end{proof}

\newcommand{\ProcessContext}{\ProcessContextC}
\newcommand{\ProcessContextC}{\mathcal{C}}
\newcommand{\ProcessContextD}{\mathcal{D}}
\newcommand{\Hole}{\bracks~}

\newcommand{\Resolve}[1]{\lfloor#1\rfloor}

\section{Soundness of the Type System}
\label{sec:proof-soundness}

\subsection{Proximity Lemma and Multicuts}

The cut elimination result of \muMALL is proved by introducing a \emph{multicut}
rule that collapses several unguarded cuts in a single rule having a variable
number of premises. The usefulness of working with multicuts is that they
prevent infinite sequences of reductions where two cuts are continuously
permuted in a non-productive way and no real progress is made in the cut
elimination process. In the type system of \piLIN we do not have a typing rule
corresponding to the multicut. At the same time, the troublesome permutations of
cut rules correspond to applications of the associativity law of parallel
composition, namely of the \refrule{s-assoc} pre-congruence rule. That is, the
introduction of the multicut rule in the cut elimination proof of \muMALL
corresponds to working with \piLIN terms considered equal up to structural
pre-congruence. Nonetheless, since the \piLIN reduction rules require two
processes willing to interact on some channel $x$ to be the children of the same
application of \refrule\CutRule, it is not entirely obvious that \emph{not}
introducing an explicit multicut rule allows us to perform in \piLIN all the
principal reductions that are performed during the cut elimination process in a
\muMALL proof.

Here we show that this is actually the case by proving a so-called
\emph{proximity lemma}, showing that every well-typed process can be rewritten
in a structurally pre-congruent one where any two processes in which the same
channel $x$ occurs free are the children of a cut on $x$.
To this aim, we introduce \emph{process contexts} to refer to unguarded
sub-terms of a process. Process contexts are processes with a single unguarded
\emph{hole} $\Hole$ and are generated by the following grammar:
\[
    \textbf{Process context} \qquad
    \ProcessContextC, \ProcessContextD ~~::=~~
    \Hole ~~\mid~~
    \Cut\x\ProcessContext{P} ~~\mid~~
    \Cut\x{P}\ProcessContext
\]

As usual we write $\ProcessContext[P]$ for the process obtained by replacing the
hole in $\ProcessContext$ with $P$. Note that this substitution is not
capture-avoiding in general and some free names occurring in $P$ may be captured
by binders in $\ProcessContext$. We write $\bn\cdot$ for the function
inductively defined by
\[
  \bn\Hole = \emptyset
  \text{\qquad and \qquad}
  \bn{\Cut\x\ProcessContext{P}} = \bn{\Cut\x{P}\ProcessContext} = \set\x \cup \bn\ProcessContext
\]

The next lemma shows that a cut on $x$ can always be pushed down towards any
process in which $x$ occurs free.

\begin{lemma}
    \label{lem:push-down}
    If $x\in\fn{P} \setminus \bn\ProcessContextC$, then $\Cut\x{\ProcessContextC[P]}{Q} \pcong
    \ProcessContextD[\Cut\x{P}{Q}]$ for some $\ProcessContextD$.
  \end{lemma}
\begin{proof}
    By induction on the structure of $\ProcessContextC$ and by cases on its
    shape.

    \begin{description}
        \item[Case $\ProcessContextC = \Hole$.]
        We conclude by taking $\ProcessContextD \eqdef \Hole$ using the
        reflexivity of $\pcong$. 

        \item[Case $\ProcessContextC = \Cut\y{\ProcessContextC'}{R}$.]
        From the hypothesis $x \in \fn{P} \setminus \bn\ProcessContextC$ we
        deduce $x \ne y$ and $x \in \fn{P} \setminus \bn{\ProcessContextC'}$.
        Using the induction hypothesis we deduce that there exists
        $\ProcessContextD'$ such that $\Cut\x{\ProcessContextC'[P]}{Q} \pcong
        \ProcessContextD'[\Cut\x{P}{Q}]$. Take $\ProcessContextD \eqdef
        \Cut\y{\ProcessContextD'}{R}$. We conclude
        \[
            \begin{array}{rcll}
                \Cut\x{\ProcessContextC[P]}{Q}
                & = & \Cut\x{\Cut\y{\ProcessContextC'[P]}{R}}{Q} & \text{by definition of $\ProcessContextC$}
                \\
                & \pcong & \Cut\x{Q}{\Cut\y{\ProcessContextC'[P]}{R}} & \text{by \refrule{s-comm}}
                \\
                & \pcong & \Cut\y{\Cut\x{Q}{\ProcessContextC'[P]}}{R} & \text{from $x \in \fn{P} \setminus \bn\ProcessContextC$ and \refrule{s-assoc}}
                \\
                & \pcong & \Cut\y{\Cut\x{\ProcessContextC'[P]}{Q}}{R} & \text{by \refrule{s-comm}}
                \\
                & \pcong & \Cut\y{\ProcessContextD'[\Cut\x{P}{Q}]}{R} & \text{using the induction hypothesis}
                \\
                & = & \ProcessContextD[\Cut\x{P}{Q}] & \text{by definition of $\ProcessContextD$}
            \end{array}
        \]

        \item[Case $\ProcessContextC = \Cut\y{R}{\ProcessContextC'}$.]
        Analogous to the previous case.
        \qedhere
    \end{description}
\end{proof}

The proximity lemma is then proved by applying \cref{lem:push-down} twice.

\begin{lemma}[proximity]
    \label{lem:proximity}
    If $x\in\fn{P_i}\setminus\bn{\ProcessContextC_i}$ for $i=1,2$, then
    $\ProcessContextC[\Cut\x{\ProcessContextC_1[P_1]}{\ProcessContextC_2[P_2]} \pcong
    \ProcessContextD[\Cut\x{P_1}{P_2}]$ for some $\ProcessContextD$.
\end{lemma}
\begin{proof}
    It suffices to apply \cref{lem:push-down} twice, thus:
    \[
        \begin{array}{@{}r@{~}c@{~}ll@{}}
            \ProcessContextC[\Cut\x{\ProcessContextC_1[P_1]}{\ProcessContextC_2[P_2]}
            & \pcong & \ProcessContextC[\ProcessContextD_1[\Cut\x{P_1}{\ProcessContextC_2[P_2]}]]
            & \text{using the hypothesis and \cref{lem:push-down}}
            \\
            & \pcong & \ProcessContextC[\ProcessContextD_1[\Cut\x{\ProcessContextC_2[P_2]}{P_1}]]
            & \text{by \refrule{s-assoc}}
            \\
            & \pcong & \ProcessContextC[\ProcessContextD_1[\ProcessContextD_2[\Cut\x{P_2}{P_1}]]]
            & \text{using the hypothesis and \cref{lem:push-down}}
            \\
            & \pcong & \ProcessContextC[\ProcessContextD_1[\ProcessContextD_2[\Cut\x{P_1}{P_2}]]]
            & \text{by \refrule{s-assoc}}
        \end{array}
    \]
    and we conclude by taking $\ProcessContextD \eqdef \ProcessContextC[\ProcessContextD_1[\ProcessContextD_2]]$.
\end{proof}

\subsection{Auxiliary Properties of Well-Typed Processes}

Most of the soundness proof of the type system relies on the cut elimination
result of \muMALL. Here we gather some auxiliary definitions and properties of
well-typed processes.
First of all, we define a function that makes a ``fair choice'' among two
processes $P$ and $Q$, by selecting the one with smaller rank. If $P$ and $Q$
happen to have the same rank, we choose $P$ by convention.

\begin{definition}
    \label{def:fc}
    The \emph{fair choice} among $P$ and $Q$, written $\fc{P}{Q}$, is defined by
    \[
        \fc{P}{Q} \eqdef
        \begin{cases}
            P & \text{if $\rankof{P} \leq \rankof{Q}$} \\
            Q & \text{otherwise}
        \end{cases}
    \]
\end{definition}

From its definition we have $\rankof{\fc{P}{Q}} = \min\set{ \rankof{P},
\rankof{Q} }$.
Next, we show that in a well-typed process there cannot be an infinite sequence
of choices if we follow the fair ones. To this aim, we define a total function
on processes that computes the length of the longest chain of subsequent fair
choices.

\begin{definition}
    \label{def:depth}
    Let $\depth\cdot$ be the function from processes to $\RankSet$ such that
    \[
        \depth{P} =
        \begin{cases}
            1 + \depth{\fc{P_1}{P_2}} & \text{if $P = \Choice{P_1}{P_2}$} \\
            0 & \text{otherwise}
        \end{cases}
    \]
\end{definition}

Note that $\depth{\fc{P}{Q}} < 1 + \depth{\fc{P}{Q}} = \depth{\Choice{P}{Q}}$.
We prove that, for well-typed processes, $\depth\cdot$ always yields a natural
number. That is, in a well-typed process there is no infinite chain of fair
choices.

\begin{lemma}
    \label{lem:depth}
    If $\wtp\Context{P}$ then $\depth{P} \in \Nat$.
\end{lemma}
\begin{proof}
    Suppose that $\depth{P} = \infty$. Then the derivation for $\qtp\Context{P}$
    has an infinite branch $\gamma = (\qtp\Context{P_i})_{i\in\omega}$ solely
    consisting of choices such that $\rankof{P_{i+1}} < \rankof{P_i}$ for every
    $i\in\omega$. Therefore, $\rankof{P_i} = \infty$ for every $i\in\omega$ and
    $\gamma$ is fair.
    From the hypothesis that $P$ is well typed we deduce that $\gamma$ is also
    valid (\cref{def:valid-branch}), namely it has a non-stationary
    $\tnu$-thread. This contradicts the fact that the contexts in $\gamma$ are
    all equal to $\Context$.
    We conclude that $\depth{P} \in \Nat$.
\end{proof}

Now we define a function $\Resolve\cdot$ on well-typed processes to statically
and fairly resolve all the choices of a process.

\begin{definition}
    \label{def:resolve}
    Let $\Resolve\cdot$ be the function on well-typed processes such that
    $\Resolve{\Choice{P}{Q}} = \Resolve{\fc{P}{Q}}$ and extended
    homomorphically to all the other process forms.
\end{definition}

The fact that $\Resolve{P}$ is uniquely defined when $P$ is well typed is a
consequence of \cref{lem:depth}. Indeed, any branch of $P$ that contains
infinitely many choices also contains infinitely many forms other than choices.
Note that the range of $\Resolve\cdot$ only contains zero-ranked processes.

The next two results prove that $\Resolve{P}$ is well typed if so is $P$.

\begin{lemma}
    \label{lem:resolved-quasi-typed}
    If $\wtp\Context{P}$ then $\qtp\Context{\Resolve{P}}$.
\end{lemma}
\begin{proof}
    \newcommand\rrel{\mathcal{R}}
    We apply the coinduction principle to show that every judgment in the set
    \[
        \rrel \eqdef \set{\qtp\Context{\Resolve{P}} \mid
        \wtp\Context{P}}
    \]
    is the conclusion of a rule in \cref{tab:typing-rules} whose premises are
    also in $\rrel$.
    Let $\qtp\Context{Q} \in \rrel$. Then $Q = \Resolve{P}$ for some $P$ such
    that $\wtp\Context{P}$.
    From \cref{lem:depth} we deduce that $\depth{P} \in \Nat$.
    We reason by induction on $\depth{P}$ and by cases on the shape of $P$ to
    show that $\qtp\Context{Q}$ is the conclusion of a rule in
    \cref{tab:typing-rules} whose premises are also in $\rrel$. We only discuss
    a few cases, the others being similar or simpler.
    
    \begin{description}
        \item[Case $P = \Cut\x{P_1}{P_2}$.]
        Then $Q = \Resolve{P} = \Cut\x{\Resolve{P_1}}{\Resolve{P_2}}$.
        From \refrule{\CutRule} we deduce that there exists $\Context_1$,
        $\Context_2$, $S_1$ and $S_2$ such that $\wtp{\Context_i, x : S_i}{P_i}$
        for $i=1,2$ and $\Context = \Context_1, \Context_2$ and $S_1 =
        \dual{S_2}$.
        Then $\qtp{\Context_i, x : S_i}{\Resolve{P_i}} \in \rrel$ by definition
        of $\rrel$ and we conclude observing that $\qtp\Context{Q} \in \rrel$ is
        the conclusion of \refrule{\CutRule}.
    
        \item[Case $P = \Choice{P_1}{P_2}$.]
        Then $Q = \Resolve{P} = \Resolve{\fc{P_1}{P_2}}$. From
        \refrule{\ChoiceRule} we deduce $\wtp\Context{\fc{P_1}{P_2}}$. Since
        $\depth{\fc{P_1}{P_2}} < \depth{P} \in \Nat$, we conclude using the
        induction hypothesis.
        \qedhere
    \end{description}
\end{proof}

\begin{lemma}
    \label{lem:resolved-well-typed}
    If $\wtp\Context{P}$ then $\wtp\Context{\Resolve{P}}$.
\end{lemma}
\begin{proof}
    Using \cref{lem:resolved-quasi-typed} we deduce $\qtp\Context{\Resolve{P}}$.
    Now, consider an infinite branch $\gamma$ in the derivation for
    $\qtp\Context{\Resolve{P}}$ and observe that $\gamma$ is necessarily fair,
    since it does not traverse any choice.
    The branch $\gamma$ corresponds to another infinite branch $\gamma'$ in the
    derivation for $\qtp\Context{P}$ which always makes fair choices whenever it
    traverses a process of the form $\Choice{P_1}{P_2}$. The only differences
    between $\gamma$ and $\gamma'$ are the judgments for the choices in $P$ that
    have been resolved. Nonetheless, all of the contexts that occur in $\gamma'$
    also occur in $\gamma$ because \refrule{\ChoiceRule} does not affect typing
    contexts.
    If $\gamma'$ traverses infinitely many choices, then $\gamma'$ traverses
    infinitely many processes with strictly decreasing ranks, which must all be
    $\infty$. Therefore, $\gamma'$ is fair.
    Since $P$ is well typed we know that $\gamma'$ is valid. Then, the
    $\tnu$-thread that witnesses the validity of $\gamma'$ corresponds to a
    $\tnu$-thread that witnesses the validity of $\gamma$.
    We conclude that $\Resolve{P}$ is well typed.
\end{proof}

\lemconservativity*
\begin{proof}
    By \cref{lem:resolved-well-typed} we deduce
    $\wtp{x_1:S_1,\dots,x_n:S_n}{\Resolve{P}}$. Since $\rankof{\Resolve{P}} =
    0$, there are no applications of \refrule{\ChoiceRule} in the derivation for
    $\qtp{x_1:S_1,\dots,x_n:S_n}{\Resolve{P}}$. That is, this derivation
    corresponds to a \muMALL derivation and every infinite branch in it is fair.
    Since $\Resolve{P}$ is well typed, we deduce that every infinite branch in
    this derivation is valid. That is, $\vdash S_1, \dots, S_n$ is derivable in
    \muMALL.
\end{proof}

The key property of zero-ranked processes that are well typed in a context $x :
\One$ is that they are weakly terminating and they eventually reduce to
$\Close\x$.

\begin{lemma}
    \label{lem:zero-rank-weak-termination}
    If $\wtp{x:\One}{P}$ and $\rankof{P} = 0$ then $P \wred \Close\x$.
\end{lemma}
\begin{proof}
    Without loss of generality we may assume that $P$ does not contain links.
    Indeed, the axiom \refrule{\LinkRule} is admissible in
    \muMALL~\cite[Proposition 10]{BaeldeDoumaneSaurin16}, so we may assume that
    links have been expanded into equivalent (choice-free) processes. We just
    note that, if the statement holds for link-free processes, then it holds
    also in the case $P$ does contain links, the only difference being that the
    sequence of reductions that are needed to fully eliminate all the cuts
    resulting from an expanded link are replaced by a single occurrence of
    \refrule{r-link}.

    From the hypothesis that $P$ has rank 0 we deduce that $P$ does not contain
    non-determimnistic choices and the derivation of $\qtp{x:\One}{P}$ does not
    contain applications of the \refrule{\ChoiceRule} rule. So, every infinite
    branch of this derivation is fair. From the hypothesis that every infinite
    fair branch of this derivation is valid we deduce that every infinite branch
    of this derivation is valid. Then this derivation corresponds to a \muMALL
    proof of $\vdash \One$.
    Observe that every principal reduction rule of \muMALL \cite[Figure
    3.2]{Doumane17} corresponds to a reduction rule of \piLIN
    (\cref{tab:semantics}).
    Also, \cref{lem:proximity} guarantees that, whenever there are two processes
    in which the same channel name $x$ occurs free we are always able to rewrite
    the process using structural precongruence so that the two sub-processes are
    the children of a cut on $x$.
    Finally, the cut elimination result for \muMALL is proved by reducing
    bottom-most cuts, meaning that principal reductions (also called
    \emph{internal reductions}~\cite{Doumane17}) are only applied at the bottom
    of a \muMALL proof.
    Therefore, each step in which a principal reduction is applied in the cut
    elimination result of \muMALL can be mimicked by a reduction of \piLIN.
    From the cut elimination result for \muMALL \cite[Proposition
    3.5]{Doumane17} we deduce that there exists no (fair) infinite sequence of
    principal reductions. That is, there exists a stuck $Q$ such that $P \wred
    Q$. From \cref{lem:quasi-subject-reduction} we deduce $\qtp{x:\One}{Q}$.
    Since the only cut-free \muMALL proof of $\vdash \One$ consists of a single
    application of \refrule[close]{\CloseRule}, and since this proof corresponds
    to the proof that $Q$ is quasi typed, we conclude that $Q = \Close\x$.
\end{proof}

\subsection{Proof of Theorem \ref{thm:soundness}}

The last auxiliary result for proving \cref{lem:weak-termination} is to show
that $P$ weakly simulates $\Resolve{P}$, namely that every reduction of a
process $\Resolve{P}$ can be mimicked by one or more reductions of $P$.

\begin{lemma}
    \label{lem:resolved-red}
    If $\wtp\Context{P}$ and $\Resolve{P} \red Q$ then $P \wred R$ for some $R$
    such that $Q = \Resolve{R}$.
\end{lemma}
\begin{proof}
    From \cref{lem:depth} we deduce $\depth{P} \in \Nat$. We proceed by double
    induction on $\depth{P}$ and on the derivation of $\Resolve{P} \red Q$ and
    by cases on the last rule applied in the derivation of $\qtp\Context{P}$.
    Most cases are straightforward since the structure of $\Resolve{P}$ and that
    of $P$ only differ for non-deterministic choices, so we only discuss the
    case \refrule{\ChoiceRule} in which $P = \Choice{P_1}{P_2}$. Then
    $\qtp\Context{\fc{P_1}{P_2}}$ and $\Resolve{P} = \Resolve{\fc{P_1}{P_2}}
    \red Q$. Recall that $\depth{\fc{P_1}{P_2}} < \depth{P} \in \Nat$. Using the
    induction hypothesis we deduce $\fc{P_1}{P_2} \wred R$ for some $R$ such
    that $Q = \Resolve{R}$. We conclude by observing that $P \red \fc{P_1}{P_2}
    \wred R$.
\end{proof}

\cref{lem:resolved-red} can be easily generalized to arbitrary sequences of
reductions.

\begin{lemma}
    \label{lem:resolved-reds}
    If $\wtp\Context{P}$ and $\Resolve{P} \wred Q$ then $P \wred R$ for some $R$
    such that $Q = \Resolve{R}$.
\end{lemma}
\begin{proof}
    A simple induction on the number of reductions in $\Resolve{P} \wred Q$
    using \cref{lem:resolved-red}.
\end{proof}

\lemweaktermination*
\begin{proof}
    Using \cref{lem:resolved-well-typed} we deduce $\wtp{x:\One}{\Resolve{P}}$.
    Using \cref{lem:zero-rank-weak-termination} we deduce $\Resolve{P} \wred
    \Close\x$. Using \cref{lem:resolved-reds} and
    \cref{lem:quasi-subject-reduction} we conclude $P \wred \Close\x$.
\end{proof}

\thmsoundness*
\begin{proof}
    By induction on the number of reductions in $P \wred Q$, from the hypothesis
    $\wtp{x : \One}{P}$ and \cref{thm:subject-reduction} we deduce $\wtp{x :
    \One}{Q}$. We conclude using \cref{lem:weak-termination}.
\end{proof}

\corfairtermination*
\begin{proof}
    Immediate consequence of \cref{thm:soundness,thm:fair-termination}.
\end{proof}

\section{Decidability of Proof Validity}
\label{sec:proof-decidability}

In this section we sketch an algorithm for deciding our notion of proof validity
(\cref{def:wtp}). Clearly, this result holds assuming that the typing derivation
for a given process is ``sufficiently regular'', namely that it is representable
as a \emph{circular derivation} consisting of finitely many distinct sub-trees.
This requires the introduction of a \emph{renaming rule}~\cite{Doumane17} to
account for the fact that addresses in types increase monotonically as they are
unfolded along branches of a derivation. In this case, the decidability of the
standard validity condition for \muMALL has been established by
Doumane~\cite[Chapter 2]{Doumane17}. The algorithm relies on two word automata,
a B\"uchi automaton $M$ that accepts all the infinite branches in the graph
representing the circular proof, and a parity automaton $N$ accepting only the
valid ones. The validity of the proof is then reduced to checking that $L(M)
\subseteq L(N)$. The difference between valid \muMALL proofs and valid \piLIN
derivations is that we only want to check this inclusion considering the
\emph{fair branches} in a \piLIN derivation, so decidability boils down to the
ability of recognizing these branches.
Considering that the fair branches are those that traverse finitely-ranked
choices finitely many times, we build a B\"uchi automaton $U$ accepting all the
unfair branches, those that traverse finitely-ranked choices infinitely many
times. The construction of $U$ is analogous to that of $M$ \cite{Doumane17},
with the difference that the set of accepting states coincides with the set of
finitely-ranked choices. From $M$ and $U$ and the fact that $\omega$-regular
languages are closed by complementation and intersection it is now possible to
construct a B\"uchi automaton that accepts $L(M) \setminus L(U)$. This language
consists of all the infinite, fair branches. Now, the validity of the \piLIN
derivation boils down to checking the inclusion $(L(M)\setminus L(U)) \subseteq
L(N)$, which is a decidable problem.

\section{Supplement to Section \ref{sec:language}}
\label{sec:extra-language}

\subsection{Finite Representation of Regular Processes}
\label{sec:infinite-processes}

An alternative, possibly more conventional way of representing infinite but
regular processes consists of considering the form $\Call{A}{x_1,\dots,x_n}$ as
being part of the syntax of processes. Then, one adds a structural precongruence
rule $\defrule{s-call}$ such that
\[
    \Call{A}{x_1,\dots,x_n} \pcong P
    \text{\qquad if \qquad}
    A(x_1,\dots,x_n) = P
\]
so that occurrences of $\Call{A}{x_1,\dots,x_n}$ can be unfolded into their
definition.
The typing rule for regular processes is the following:
\[
    \inferrule{
        \qtp{x_1:S_1,\dots,x_n:S_n}P
    }{
        \qtp{x_1:S_1,\dots,x_n:S_n}{\Call{A}{x_1,\dots,x_n}}
    }
    ~\defrule{call}
\]

With these changes, the only difference in the presented results concerns the
formulation of \cref{thm:soundness}, since a process $\Call{A}{x}$ when $A(x) =
\Close\x$ does \emph{not} reduce and yet it is not \emph{syntactically equal} to
$\Close\x$. The solution is to slightly weaken the statement of the theorem
allowing the use of structural pre-congruence for rewriting the final state of
$Q$:

\begin{theorem}
    If $\wtp{x : \One}{P}$ and $P \wred Q$ then $Q \wred\pcong \Close\x$.
\end{theorem}

\subsection{On the Definition of Rank}
\label{sec:extra-rank}

Here we provide a more detailed justification of the fact that the notion of
rank (\cref{def:rank}) is well defined. As we have already observed, the value
of $\rankof{P}$ only depends on the structure of $P$, but not on the actual
names occurring in $P$. As a consequence, when $P$ is defined by means of a
\emph{finite} system of equations, the value of $\rankof{P}$ too can be
determined by a \emph{finite} system of equations, even though $P$ is not a
regular tree strictly speaking.

As an example, consider the definition of $\Buyer$ found in
\cref{ex:buyer-seller}. In order to compute $\rankof{\Call\Buyer\x}$ we consider
the system of equations
{
    \renewcommand{\x}{\bullet}
    \renewcommand{\y}{\bullet}
    \renewcommand{\z}{\bullet}
    \begin{align*}
        \rankof{\Call\Buyer\x} & = \rankof{
            \Choice{
                \Select[\z]\AddTag\y.\Call\Buyer\z
            }{
                \Select[\z]\PayTag\y.\Close\z
            }
        }
        \\
        \rankof{
            \Choice{
                \Select[\z]\AddTag\y.\Call\Buyer\z
            }{
                \Select[\z]\PayTag\y.\Close\z
            }
        } & = 1 + \min\set{
            \rankof{\Select[\z]\AddTag\y.\Call\Buyer\z},
            \rankof{\Select[\z]\PayTag\y.\Close\z}
        }
        \\
        \rankof{\Select[\z]\AddTag\y.\Call\Buyer\z} & = \rankof{\Call\Buyer\z}
        \\
        \rankof{\Select[\z]\PayTag\y.\Close\z} & = \rankof{\Close\z}
        \\
        \rankof{\Close\z} & = 0
    \end{align*}
}
where we have used a placeholder $\bullet$ in place of every channel name
occurring in these terms.

\newcommand{\fun}{\mathcal{F}}

Every such system of equations can be thought of as a function $\fun :
\parens\RankSet^n \to \parens\RankSet^n$ on the complete lattice
$\parens\RankSet^n$ ordered by the pointwise extension of $\leq$ in $\RankSet$.
Note that $\fun$ is monotone, because all the operators occurring in the
definition of rank (\cref{def:rank}) are monotone. So, $\fun$ has a least fixed
point by the Knaster-Tarski theorem and the rank of $P$ is the component of this
fixed point that corresponds to $\rankof{P}$.

For the above system of equations we have $n = 5$ and we have
\[
    \fun(x_1,x_2,x_3,x_4,x_5) = (x_2,1+\min\set{x_3,x_4},x_1,x_5,0)
\]
whose least solution is $(1,1,1,0,0)$. Now $\rankof{\Call\Buyer\x}$ corresponds
to the first component of this solution, that is $1$.

\section{Typing Derivations for Example \ref{ex:slot-machine}}
\label{sec:extra-soundness}

Below is the typing derivation showing that $\Player$ is quasi typed. For
convenience, we define $\Formula' \eqdef (\Formula \branch \Formula) \choice
\One$. Note that $\Formula'$ is the unfolding of $\Formula$ defined in
\cref{ex:slot-machine}.

\[
    \begin{prooftree}
        \[
            \[
                \[
                    \[
                        \mathstrut\smash\vdots
                        \justifies
                        \qtp{
                            x : \FormulaF_{aill}
                        }{
                            \Call\Player\x
                        }
                    \]
                    \qquad
                    \[
                        \[
                            \[
                                \justifies
                                \qtp{
                                    x : \One
                                }{
                                    \Close\x
                                }
                                \using\refrule[close]{\CloseRule}
                            \]
                            \justifies
                            \qtp{
                                x : \Formula'_{ailri}
                            }{
                                \Select\QuitTag\x.\Close\x
                            }
                            \using\refrule[select]{\SelectRule}
                        \]
                        \justifies
                        \qtp{
                            x : \FormulaF_{ailr}
                        }{
                            \Rec\x.\Select\QuitTag\x.\Close\x
                        }
                        \using\refrule[rec]{\RecRule}
                    \]
                    \justifies
                    \qtp{
                        x : (\FormulaF \branch \FormulaF)_{ail}
                    }{
                        \Case\x{
                            \Call\Player\x
                        }{
                            \Rec\x.\Select\QuitTag\x.\Close\x
                        }    
                    }
                    \using\refrule[case]{\CaseRule}
                \]
                \justifies
                \qtp{
                    x : \Formula'_{ai}
                }{
                    \Select\PlayTag\x.
                    \Case\x{
                        \Call\Player\x
                    }{
                        \Rec\x.\Select\QuitTag\x.\Close\x
                    }
                }
                \using\refrule[select]{\SelectRule}
            \]
            \[
                \[
                    \justifies
                    \qtp{
                        x : \One
                    }{
                        \Close\x
                    }
                    \using\refrule[close]{\CloseRule}
                \]
                \justifies
                \qtp{
                    x : \Formula'_{ai}
                }{
                    \Select\QuitTag\x.\Close\x
                }
                \using\refrule[select]{\SelectRule}
            \]
            \justifies
            \qtp{
                x : \Formula'_{ai}
            }{
                \Choice{
                    \Select\PlayTag\x.
                    \Case\x{
                        \Call\Player\x
                    }{
                        \Rec\x.\Select\QuitTag\x.\Close\x
                    }
                }{
                    \Select\QuitTag\x.\Close\x
                }
            }
            \using\refrule{\ChoiceRule}
        \]
        \justifies
        \qtp{
            x : \Formula_a
        }{
            \Call\Player\x
        }
        \using\refrule[rec]{\RecRule}
    \end{prooftree}
\]

Below is the typing derivation showing that $\Machine$ is quasi typed. For
convenience, we define $\FormulaG' \eqdef \FormulaG \choice \FormulaG$ and we
abbreviate $\Machine$ to $M$.

\[
    \renewcommand{\Machine}{M}
    \begin{prooftree}
        \[
            \[
                \[
                    \[
                        \mathstrut\smash\vdots
                        \justifies
                        \qtp{
                            x : \FormulaG_{bill},
                            y : \One
                        }{
                            \Call\Machine{x,y}
                        }
                    \]
                    \justifies
                    \qtp{
                        x : \FormulaG'_{bil},
                        y : \One        
                    }{
                        \Select\WinTag\x.\Call\Machine{x,y}
                    }
                    \using\refrule[select]\SelectRule
                \]
                \[
                    \[
                        \mathstrut\smash\vdots
                        \justifies
                        \qtp{
                            x : \FormulaG_{bilr},
                            y : \One            
                        }{
                            \Call\Machine{x,y}
                        }
                    \]
                    \justifies
                    \qtp{
                        x : \FormulaG'_{bil},
                        y : \One        
                    }{
                        \Select\LoseTag\x.\Call\Machine{x,y}
                    }
                    \using\refrule[select]\SelectRule
                \]
                \justifies
                \qtp{
                    x : \FormulaG'_{bil},
                    y : \One    
                }{
                    \Choice{
                        \Select\WinTag\x.\Call\Machine{x,y}
                    }{
                        \Select\LoseTag\x.\Call\Machine{x,y}
                    }
                }
                %\using\refrule\ChoiceRule
            \]
            \hspace{-1em}
            \[
                \[
                    \justifies
                    \qtp{
                        y : \One
                    }{
                        \Close\y
                    }
                    \using\refrule[close]{\CloseRule}
                \]
                \justifies
                \qtp{
                    x : \Bot,
                    y : \One
                }{
                    \Wait\x.\Close\y    
                }
                \using\refrule[wait]{\WaitRule}
            \]
            \justifies
            \qtp{
                x : (\FormulaG' \branch \Bot)_{bi},
                y : \One
            }{
                \Case\x{
                    \Choice{
                        \Select\WinTag\x.\Call\Machine{x,y}
                    }{
                        \Select\LoseTag\x.\Call\Machine{x,y}
                    }
                }{
                    \Wait\x.\Close\y
                }            
            }
            \using\refrule[case]{\CaseRule}
        \]
        \justifies
        \qtp{
            x : \FormulaG_b,
            y : \One
        }{
            \Call\Machine{x,y}
        }
        \using\refrule[corec]{\CorecRule}
    \end{prooftree}
\]

\newcommand{\Sender}{\textit{Sender}}
\newcommand{\Receiver}{\textit{Receiver}}
\newcommand{\LeafTag}{\mktag{leaf}}
\newcommand{\BranchTag}{\mktag{branch}}

\section{Encoding of a Context-Free Protocol}
\label{sec:context-free-protocol}

In this appendix we illustrate the encoding of a \emph{context-free
communication protocol} in \piLIN. A context-free protocol (or session type)
cannot be described accurately by plain session types in which recursion is
allowed in tail position only. A paradigmatic example of such a protocol is that
describing the (de)serialization of a binary
tree~\cite{ThiemannVasconcelos16,Padovani19}. Below we model two processes that
exchange a binary tree.
\[
    \begin{array}{@{}r@{~}c@{~}l@{}}
        \Sender(x) & = &
        \Rec\x.(\Choice{
            \Select\LeafTag\x.\Close\x
        }{
            \Select\BranchTag\x.
            \Fork[z]\x\y{
                \Call\Sender\y
            }{
                \Call\Sender\z
            }
        })
        \\
        \Receiver(x,y) & = &
        \Corec\x.
        \Case\x{
            \Wait\x.\Close\y
        }{
            \Join[z]\x\y.
            \Cut{u}{
                \Call\Receiver{y,u}
            }{
                \Wait{u}.\Call\Receiver{z,y}
            }
        }
    \end{array}
\]

The $\Sender$ performs a non-deterministic choice corresponding to the structure
of the tree to be communicated: if the tree is empty (that is a $\LeafTag$), it
sends the unit indicating that no further data will be transmitted; if the tree
is not empty (that is a $\BranchTag$), it sends a pair of channels along which
the left and right subtrees are communicated in parallel. The $\Receiver$
behaves in a dual way. After receiving $\LeafTag$, it signals the completion of
the communication by sending the unit on $y$. After receiving $\BranchTag$, it
first receives the left subtree followed by the right subtree. The sequencing
between these two phases is enforced by the communication on the new channel
$u$.

It is easy to see that $\Sender$ uses $x$ according to the formula $\FormulaF
\eqdef \tmu\X.\One \choice (X \tfork X)$ and that $\Receiver$ uses $x$ and $y$
according to $\FormulaG \eqdef \tnu\X.\Bot \branch (X \tjoin X)$ and $\One$,
respectively. In particular, the $\tjoin$ connective in $\FormulaG$ plays the
same role of sequential composition in context-free session
types~\cite{ThiemannVasconcelos16}.

\end{document}